\documentclass{lmcs}
\pdfoutput=1

\usepackage{lastpage}
\lmcsdoi{17}{2}{5}
\lmcsheading{}{\pageref{LastPage}}{}{}%
{Jun.~02,~2020}{Apr.~14,~2021}{}

\usepackage[utf8]{inputenc}


\usepackage{mysty}
\usepackage{wrapfig}

\title{Sculptures in Concurrency}

\author[U.\ Fahrenberg]{Uli Fahrenberg\rsuper{a}}%
\address{\lsuper{a}{\'E}cole Polytechnique, Palaiseau, France}

\author[C.\ Johansen]{Christian Johansen\rsuper{b}}%
\address{\lsuper{b}Norwegian University of Science and Technology, Gj\o{}vik, Norway}

\author[C.A.\ Trotter]{Christopher A. Trotter\rsuper{c}}%
\address{\lsuper{c}Institute of Informatics, University of Oslo, Oslo, Norway}

\author[K.\ Ziemia{\'n}ski]{Krzysztof Ziemia{\'n}ski\rsuper{d}}%
\address{\lsuper{d}Faculty of Mathematics, Informatics and Mechanics, University of Warsaw, Warsaw, 
Poland}

\begin{document}

\setcounter{footnote}0

\begin{abstract}
  We give a formalization of Pratt's intuitive sculpting process for
  higher-dimensional automata (HDA). 
Intuitively, an HDA is a sculpture if it can be embedded in (\ie~sculpted from) a single higher dimensional cell (hypercube).
A first important result of this paper is that not all HDA can be sculpted, exemplified through several natural acyclic HDA, one being the famous ``broken box'' example of van Glabbeek. Moreover, we show that even the natural operation of unfolding is completely unrelated to sculpting, e.g., there are sculptures whose unfoldings cannot be sculpted.
We investigate the expressiveness of sculptures, as a proper subclass of HDA, by showing them to be equivalent to regular ST-structures (an event-based counterpart of HDA) and to (regular) Chu spaces over 3 (in their concurrent interpretation given by Pratt).
We believe that our results shed new light on the intuitions behind sculpting as a method of modeling concurrent behavior, showing the precise reaches of its expressiveness. 
Besides expressiveness, we also develop an algorithm to decide whether an HDA can be sculpted.
More importantly, we show that sculptures are equivalent to Euclidean cubical complexes (being the geometrical counterpart of our combinatorial definition), which include the popular PV models used for deadlock detection. This exposes a close connection between
  geometric and combinatorial models for concurrency which may be of
  use for both areas.
\end{abstract}

\maketitle

\section{Introduction}

In approaches to non-interleaving concurrency, more than one event may
happen simultaneously.  There is a plethora of formalisms for modeling
and analyzing such concurrent systems, \eg~Petri
nets~\cite{book/Petri62}, event
structures~\cite{DBLP:journals/tcs/NielsenPW81}, configuration
structures~\cite{DBLP:journals/tcs/GlabbeekP09, GlabbeekP95config},
or more recent variations such as dynamic event
structures~\cite{DBLP:conf/forte/ArbachKPN15} or
ST-structures~\cite{Johansen16STstruct, P12turing}.
They all share the idea of differentiating between concurrent and
interleaving executions; \ie~in CCS notation~\cite{book/Milner89},
$a\,|\,b$ is not the same as $ a. b+ b. a$.

In~\cite{DBLP:journals/tcs/Glabbeek06}, van~Glabbeek shows that (up to
history-preserving bisimilarity) \emph{higher-dimensional automata}
(HDA), introduced by Pratt and van~Glabbeek
in~\cite{DBLP:conf/popl/Pratt91, Glabbeek91-hda}, encompass all other
commonly used models for concurrency.
However, their generality make HDA  
quite difficult to work with, and so the quest for useful and general
models for concurrency continues.

In~\cite{Pratt00Sculptures}, Pratt introduces \emph{sculpting} as a
process to manage the complexity of HDA. 
Intuitively, sculpting takes one single hypercube, having enough
concurrency (\ie~enough events), and removes cells until the desired
concurrent behavior is obtained.  This is
orthogonal to \emph{composition}, where a system is built by putting
together smaller systems, which in HDA is done by gluing cubes.  Pratt
finishes the introduction of~\cite{Pratt00Sculptures} saying that
``sculpture on its own suffices [\dots] for the abstract modeling of
concurrent behavior.''

\begin{figure}[t]
  \centering
  \includegraphics[width=.5\linewidth, trim=0 4ex 0 4ex, clip]{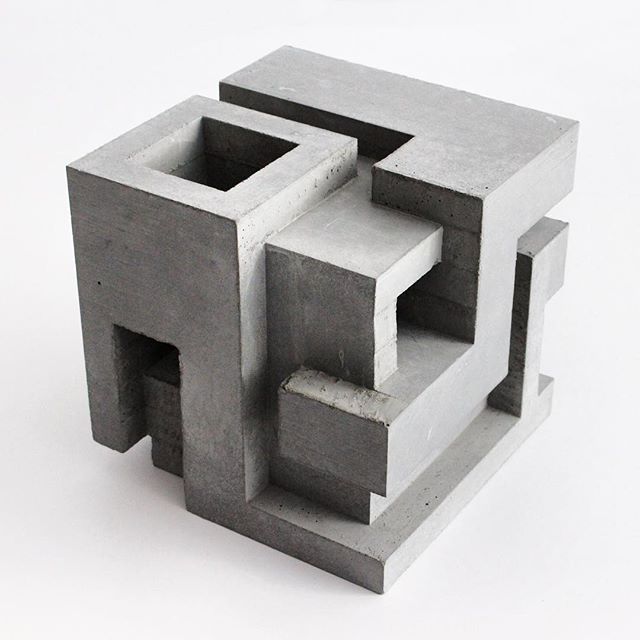}
  \caption{%
    \label{fig:sculpt}
    A geometric sculpture: David Umemoto, \textit{Cubic Geometry ix-vi}.
  }
\end{figure}

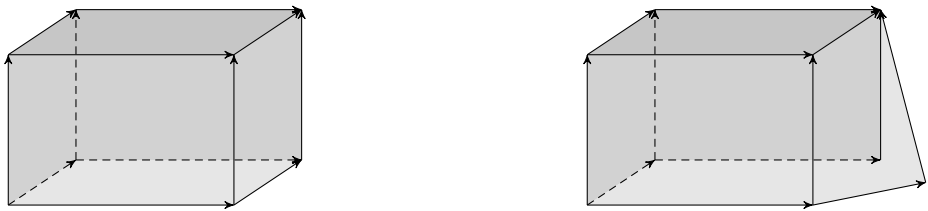
\begin{figure}[b]
  \centering
  \begin{tikzpicture}[>=stealth', x=3cm, y=2cm]
    \begin{scope}
      \coordinate (000) at (0,0);
      \coordinate (001) at (1,0);
      \coordinate (010) at (.3,.3);
      \coordinate (100) at (0,1);
      \coordinate (011) at (1.3,.3);
      \coordinate (101) at (1,1);
      \coordinate (110) at (.3,1.3);
      \coordinate (111) at (1.3,1.3);
      \fill[gray!20] (000) -- (001) -- (011) -- (010) -- (000);
      \fill[gray!35] (000) -- (010) -- (011) -- (111) -- (101) --
      (100) -- (000);
      \fill[gray!45] (100) -- (101) -- (111) -- (110) -- (100);
      \path (000) edge (001);
      \path (000) edge[densely dashed] (010);
      \path (000) edge (100);
      \path (001) edge (011);
      \path (001) edge (101);
      \path (010) edge[densely dashed] (110);
      \path (010) edge[densely dashed] (011);
      \path (100) edge (110);
      \path (100) edge (101);
      \path (011) edge (111);
      \path (101) edge (111);
      \path (110) edge (111);
    \end{scope}
    \begin{scope}[xshift=20em]
      \coordinate (000) at (0,0);
      \coordinate (001) at (1,0);
      \coordinate (010) at (.3,.3);
      \coordinate (100) at (0,1);
      \coordinate (011a) at (1.5,.15);
      \coordinate (011b) at (1.3,.3);
      \coordinate (101) at (1,1);
      \coordinate (110) at (.3,1.3);
      \coordinate (111) at (1.3,1.3);
      \fill[gray!20] (000) -- (001) -- (011a) -- (111) -- (011b) --
      (010) -- (000);
      \fill[gray!35] (000) -- (010) -- (011b) -- (111) -- (101) --
      (100) -- (000);
      \fill[gray!45] (100) -- (101) -- (111) -- (110) -- (100);
      \path (000) edge (001);
      \path (000) edge[densely dashed] (010);
      \path (000) edge (100);
      \path (001) edge (011a);
      \path (001) edge (101);
      \path (010) edge[densely dashed] (110);
      \path (010) edge[densely dashed] (011b);
      \path (100) edge (110);
      \path (100) edge (101);
      \path (011a) edge (111);
      \path (011b) edge (111);
      \path (101) edge (111);
      \path (110) edge (111);
    \end{scope}
  \end{tikzpicture}
  \caption{%
    \label{fig:boxes}
    A combinatorial sculpture, the upside-down open box, or
    ``Fahrenberg's matchbox''~\cite{DBLP:conf/icalp/DubutGG15} (left),
    and its unfolding (right), the ``broken box'' which cannot be
    sculpted (this was the example of van Glabbeek \cite[Fig.~11]{DBLP:journals/tcs/Glabbeek06}, though not named as we do).}
\end{figure}

In this paper we make precise the intuition of
Pratt~\cite{Pratt00Sculptures} and give a definition of sculptures.
We show that there is a close correspondence between sculptures, Chu
spaces over $\three$~\cite{pratt95chu}, and ST-structures.  We develop
an algorithm to decide whether an HDA can be sculpted and show in
Theorem~\ref{th:conc} several natural examples of acyclic HDA that
\emph{cannot} be sculpted.  We will carefully introduce these concepts
later, but spend some time here to motivate our developments.

Combinatorial sculpting as described above is not to be confused with
\emph{geometric} sculpting, which consists of taking a geometric cube
of some dimension and chiseling away hypercubes which one does not
want to be part of the structure.  Figure~\ref{fig:sculpt} shows a
geometric sculpture; for a combinatorial sculpture see
Figure~\ref{fig:boxes}.

Geometric sculpting has been used by Fajstrup~\etal\
in~\cite{Fajstrup06, FajstrupGR98} and other papers to model and
analyze so-called PV programs: processes which interact by locking and
releasing shared resources.  In the simplest case of linear processes
without choice or iteration this defines a hypercube with
\emph{forbidden hyperrectangles}, which cannot be accessed due
to resources' access limits.  See Figure~\ref{fig:swiss} for an
example.

Technically, geometric sculptures are \emph{Euclidean cubical
  complexes}; rewriting a proof in~\cite{Ziemianski17} we show that
such complexes are precisely (combinatorial) sculptures.  In other words, an HDA
is Euclidean iff it can be sculpted, so that the geometric models for
concurrency~\cite{Fajstrup06, FajstrupGR98} are closely related to the
combinatorial ones~\cite{DBLP:conf/popl/Pratt91, Glabbeek91-hda},
through the notion of sculptures.  Much work has been done in the
\emph{geometric} analysis of Euclidean
HDA~\cite{FajstrupRGH04FundCatI, GoubaultH07FundCatII, FajstrupGR98,
  MeshulamR17, RaussenZ14, Ziemianski17}; through our equivalences
these results are made available for the combinatorial models.

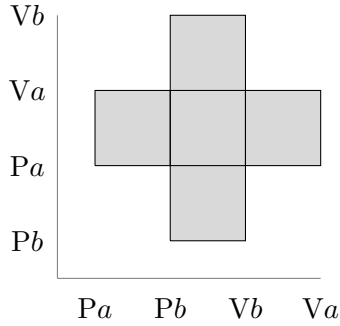
\begin{figure}[t]
    \centering
  \begin{minipage}[b]{.5\linewidth}
    \centering
  \begin{tikzpicture}[-]
    \path (.5,.5) edge[gray] (4,.5);
    \path (.5,.5) edge[gray] (.5,4);
    \draw[fill=gray!30] (1,2) -- (4,2) -- (4,3) -- (1,3) -- (1,2);
    \draw[fill=gray!30] (2,1) -- (2,4) -- (3,4) -- (3,1) -- (2,1);
    \draw (2,2) -- (3,2) -- (3,3) -- (2,3) -- (2,2);
    \node at (1,.1) {P$a$};
    \node at (2,.1) {P$b$};
    \node at (3,.1) {V$b$};
    \node at (4,.1) {V$a$};
    \node at (.1,1) {P$b$};
    \node at (.1,2) {P$a$};
    \node at (.1,3) {V$a$};
    \node at (.1,4) {V$b$};
  \end{tikzpicture}
\end{minipage}
  \caption{%
    \label{fig:swiss}
    Two PV processes sharing two mutexes.
    The forbidden area is grayed out.%
  }
\end{figure}

The notion of \emph{unfolding} 
is commonly used to turn a complicated
model into a simpler, but potentially infinite one.
It may thus be expected that even if an HDA cannot be sculpted, then at
least its
unfolding can, as illustrated by the two examples in
Figure~\ref{fig:unfoldings}.
\begin{figure}[b]
  \centering
  \vspace*{-2ex}
  \begin{tikzpicture}[>=stealth']
    \begin{scope}
      \node[state, initial] (0) at (0,0) {};
      \node[state] (1) at (1,0) {};
      \path (0) edge[out=30, in=150] node[above] {$a$} (1);
      \path (1) edge[out=210, in=-30] node[below] {$b$} (0);
    \end{scope}
    \begin{scope}[xshift=13em]
      \node[state, initial] (0) at (0,0) {};
      \foreach \x in {1,2,3}
      \node[state] (\x) at (\x,0) {};
      \node (n) at (4,0) {$\dots$};
      \foreach \x/\y in {0/1,2/3}
      \path (\x) edge node[above] {$a$} (\y);
      \foreach \x/\y in {1/2}
      \path (\x) edge node[above] {$b$} (\y);
      \path (3) edge (n);
    \end{scope}
    \begin{scope}[yshift=-12ex, xshift=-1em]
      \node[state, initial] (0) at (0,0) {};
      \node[state] (1) at (.8,.6) {};
      \node[state] (2) at (1.6,0) {};
      \path (0) edge node[above, pos=.3] {$b$} (1);
      \path (1) edge node[above, pos=.7] {$c$} (2);
      \path (0) edge node[below] {$a$} (2);
    \end{scope}
    \begin{scope}[xshift=13em, yshift=-12ex]
      \node[state, initial] (0) at (0,0) {};
      \node[state] (1) at (.8,.6) {};
      \node[state] (2a) at (1.6,.6) {};
      \node[state] (2b) at (1.6,0) {};
      \path (0) edge node[above, pos=.3] {$b$} (1);
      \path (1) edge node[above, pos=.7] {$c$} (2a);
      \path (0) edge node[below] {$a$} (2b);
    \end{scope}
  \end{tikzpicture}
  \caption{%
    \label{fig:unfoldings}
    Two simple HDA which cannot be sculpted (left) and their
    unfoldings (right) which can.  (The top-right sculpture is
    infinite.)}
\end{figure}
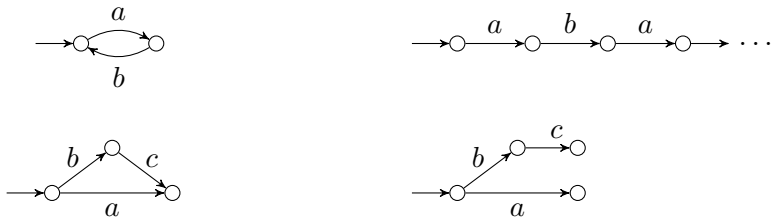
However, this is \emph{not} always the case, as witnessed by the
example in Figure~\ref{fig:speedAngelDemon} which shows an HDA
which cannot be sculpted and which is its own unfolding.
This concurrent system, introduced in~\cite{Johansen16STstruct}, cannot
be modeled as an ST-structure, but \emph{can} be modeled as an
ST-structure with
\emph{cancellation}~\cite[Sec.~5]{Johansen16STstruct}.

\begin{figure}[tbp]
  \centering
\begin{minipage}[b]{.4\linewidth}
  \centering
  \begin{tikzpicture}[x=1.5cm, y=1.2cm]
    \coordinate (00') at (0,0);
    \coordinate (10a') at (1.2,-.3);
    \coordinate (20a') at (2.4,-.7);
    \coordinate (10b') at (.9,.3);
    \coordinate (20b') at (2.1,.8);
    \coordinate (01') at (0,1);
    \coordinate (11') at (1,1.1);
    \coordinate (21a') at (2.3,.5);
    \coordinate (21b') at (2.3,1.6);

    \fill[gray!30] (00') -- (10a') -- (20a') -- (21a') -- (11') --
    (01') -- (00');
    \fill[gray!30] (00') -- (10b') -- (20b') -- (21b') -- (11') --
    (01') -- (00');

    \node[state, initial] (00) at (00') {};
    \foreach \a in {10a,20a,10b,20b,01,11,21a,21b}
    \node[state] (\a) at (\a') {};

    \path (00) edge node[below] {$d$} (10a);
    \path (10a) edge node[below] {$b$} (20a);
    \path (00) edge node[above] {$d$} (10b);
    \path (10b) edge node[below] {$c$} (20b);
    \path (01) edge node[above] {$d$} (11);
    \path (11) edge node[above] {$b$} (21a);
    \path (11) edge node[above] {$c$} (21b);
    \path (00) edge node[left] {$a$} (01);
    \path (10a) edge node[right, pos=.3] {$a$} (11);
    \path (10b) edge node[left] {$a$} (11);
    \path (20a) edge node[right] {$a$} (21a);
    \path (20b) edge node[right] {$a$} (21b);
  \end{tikzpicture}
\end{minipage}
  \caption{%
    \label{fig:speedAngelDemon}
    The speed game of angelic vs.\ demonic
    choice~\cite{Johansen16STstruct}.}
\end{figure}
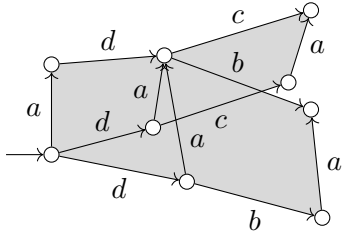

Even more concerning is the fact that there are HDA which can be
sculpted, but their unfoldings cannot; in fact, Figure~\ref{fig:boxes}
exposes one such example.  This shows that for HDA, unfolding does not
always return a simpler model.

In the geometric setting, this means that there are Euclidean cubical
complexes whose unfoldings are not Euclidean.  Since Goubault and
Jensen's seminal paper~\cite{DBLP:conf/concur/GoubaultJ92},
\emph{directed topology} has been developed in order to analyze
concurrent systems as geometric objects~\cite{Grandis09book,
  Fajstrup06, DBLP:books/sp/FajstrupGHMR16}.  Directed topology has
been developed largely in analogy to algebraic topology, but the
analogy sometimes breaks.
The mismatch we discover here, between Euclidean complexes and
unfoldings, shows such a broken analogy.  Unfoldings of HDA are
directed analogues of \emph{universal covering spaces} in
algebraic topology~\cite{Glabbeek91-hda, Fahrenberg05CatHDA,
  Uli05PhD}.
There are several other problems with this notion,
and finding better definitions of directed coverings is active ongoing
research~\cite{DBLP:conf/fossacs/Dubut19, DBLP:conf/calco/FahrenbergL15}.

Another motivation for Pratt's~\cite{Pratt00Sculptures} is that HDA
have no explicit notion of \emph{events}.  From the work
in~\cite{Johansen16STstruct} on ST-structures, %
introduced as event-based counterparts of HDA, we know that it is not
always possible to identify the events in an HDA.  The example in
Figure~\ref{fig:conflict} shows the \emph{(strong) asymmetric conflict}
from~\cite{DBLP:journals/tcs/GlabbeekP09, Pratt03trans_cancel,
  Johansen16STstruct}, with two events $a$, $b$ such that occurrence
of $a$ disables $b$. This can be modeled as a general event structure,
but not as a pure event structure, hence also not as a configuration
structure~\cite{DBLP:journals/tcs/GlabbeekP09}.  It can also be
modeled as an ST-structure, but when using HDA, one faces the problem
that HDA transition labels do not carry events.  The right part of
Figure~\ref{fig:conflict} shows two different ways of sculpting the
corresponding structure from an HDA, one in which the two $a$-labeled
transitions denote the same event and one in which they do not;
a priori there is no way to tell which HDA is the ``right'' model.
This also shows that the same HDA may be sculpted in several different
ways.

\begin{figure}[bp]
  \centering
  \begin{tikzpicture}[>=stealth']
    \begin{scope}[x=1.1cm]
      \node (00) at (0,0) {$\emptyset$};
      \node (10) at (1,0) {$\{ a\}$};
      \node (01) at (0,1) {$\{ b\}$};
      \node (11) at (1,1) {$\{ b, a\}$};
      \path (00) -- node[anchor=center] {$\vdash$} (10);
      \path (00) -- node[anchor=center, sloped] {$\vdash$} (01);
      \path (01) -- node[anchor=center] {$\vdash$} (11);
    \end{scope}
    \begin{scope}[xshift=10em, x=1.5cm, y=.8cm, yshift=-.3cm]
      \node (00) at (0,0) {$( \emptyset, \emptyset)$};
      \node (t0) at (1,0) {$( a, \emptyset)$};
      \node (10) at (2,0) {$( a, a)$};
      \node (0t) at (0,1) {$( b, \emptyset)$};
      \node (01) at (0,2) {$( b, b)$};
      \node (t1) at (.95,2) {$( ba, b)$};
      \node (11) at (2,2) {$( ba, ba)$};
      \foreach \a/\b in {00/t0,t0/10,00/0t,0t/01,01/t1,t1/11}
      \path (\a) edge (\b);
    \end{scope}
    \begin{scope}[xshift=26em, yshift=-7ex]
      \node[state, initial] (00) at (0,0) {};
      \foreach \x/\y in {0/1,1/0,1/1}
      \node[state] (\x\y) at (\x,\y) {};
      \draw[-, dotted] (00) -- (10) -- (11) -- (01) -- (00);
      \path (00) edge node[below] {$a$} (10);
      \path (00) edge node[left] {$b$} (01);
      \path (01) edge node[above] {$a$} (11);
    \end{scope}
    \begin{scope}[xshift=25em, yshift=7ex, x=1.5cm, y=1.1cm]
      \node[state, initial] (000) at (0,0) {};
      \node[state] (001) at (1,0) {};
      \node[state] (010) at (.4,.3) {};
      \node[state] (100) at (0,1) {};
      \node[state] (011) at (1.4,.3) {};
      \node[state] (101) at (1,1) {};
      \node[state] (110) at (.4,1.3) {};
      \node[state] (111) at (1.4,1.3) {};
      \draw[-, dotted] (000) -- (100) -- (110) -- (010) -- (000);
      \draw[-, dotted] (001) -- (101) -- (111) -- (011) -- (001);
      \draw[-, dotted] (000) -- (001);
      \draw[-, dotted] (100) -- (101);
      \draw[-, dotted] (110) -- (111);
      \draw[-, dotted] (010) -- (011);
      \path (000) edge node[below] {$a$} (001);
      \path (000) edge node[above, pos=.3] {$b$} (010);
      \path (010) edge node[right] {$a$} (110);
    \end{scope}
  \end{tikzpicture}
  \caption{%
    \label{fig:conflict}
    Asymmetric conflict as an (impure) event structure (left), an
    ST-structure (center), and two different interpretations as HDA
    (right).}
\end{figure}
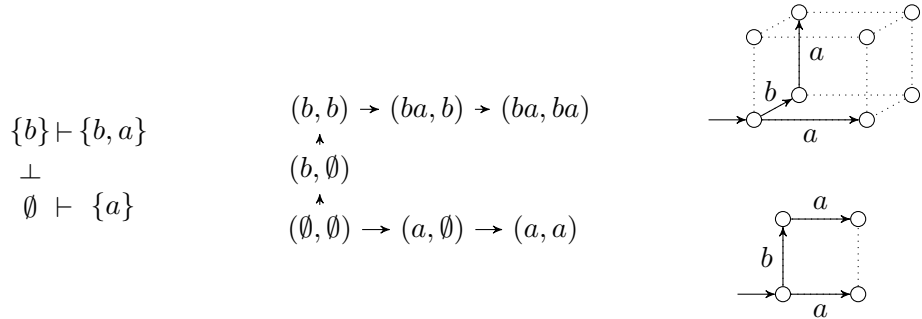

\newpage

\paragraph*{Structure of the paper}

We start in Section~\ref{sec_background} by recalling the definitions
of HDA, ST-structures, and Chu spaces.
In Section~\ref{sec_sculptures} we introduce sculptures and show that
they are isomorphic to \emph{regular} ST-structures.  The triple
equivalence
\begin{equation*}
  \text{regular ST-structures --- regular Chu spaces --- sculptures}
\end{equation*}
embodies Pratt's \emph{event-state
  duality}~\cite{Pratt92concur}.
Regularity is a geometric closure condition introduced for
ST-structures in~\cite{Johansen16STstruct} which ensures that for any
ST-configuration, also all its faces are part of the structure, and
they are all distinct.  If regularity is dropped, then one has to pass
to \emph{partial HDA}~\cite{DBLP:conf/calco/FahrenbergL15} on the
geometric side, and then the above equivalence becomes one between
ST-structures and sculptures from partial HDA.  For clarity of
exposition we do not pursue this here, but also in that case, there
will be acyclic partial HDA which cannot be sculpted.

Section~\ref{sec_decision} contains our main contribution, an
algorithm to decide whether a given HDA $Q$ can be sculpted.  The
algorithm essentially works by covering $Q$ with the ST-structure
$\hintost(Q)$ which is built out of all paths in $Q$, and then trying
to find a quotient of $\hintost(Q)$ which is isomorphic to $Q$.  We
show that such a quotient exists iff $Q$ can be sculpted.

\begin{figure}[btp]
  \centering
  \begin{tikzpicture}
    \begin{scope}[x=1.5cm, y=1.5cm]
      \node[state, initial] (00) at (0,0) {};
      \node[state] (10) at (1,0) {};
      \node[state] (01) at (0,1) {};
      \node[state] (11) at (1,1) {};
      \path (00) edge node[below] {$q_1$} (10);
      \path (00) edge node[left] {$q_2$} (01);
      \path (01) edge node[above] {$q_3$} (11);
      \path (10) edge node[right] {$q_4$} (11);
    \end{scope}
    \begin{scope}[xshift=10em, y=.9cm, x=2cm]
      \node (00) at (0,0) {$(\emptyset, \emptyset)$};
      \node (t0) at (1,0) {$( q_1, \emptyset)$};
      \node (10) at (2,0) {$( q_1, q_1)$};
      \node (1t) at (2.6,1) {$( q_1 q_4, q_1)$};
      \node (11a) at (3.2,2) {$( q_1 q_4, q_1 q_4)$};
      \node (0t) at (0,1) {$( q_2, \emptyset)$};
      \node (01) at (0,2) {$( q_2, q_2)$};
      \node (t1) at (.92,2) {$( q_2 q_3, q_2)$};
      \node (11b) at (2,2) {$( q_2 q_3, q_2 q_3)$};
      \foreach \a/\b in {00/t0,t0/10,10/1t,1t/11a,00/0t,0t/01,01/t1,t1/11b}
      \path (\a) edge (\b);
    \end{scope}
  \end{tikzpicture}
  \caption{%
    \label{fig:opensquare}
    A simple HDA and its path-based ST-structure covering.}
\end{figure}
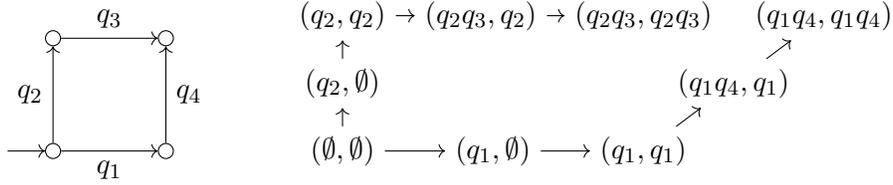

Figure~\ref{fig:opensquare} shows a simple example: the empty square,
a one-dimensional HDA with two interleaving transitions.  The covering
$\hintost(H)$ splits the upper-right corner, and the algorithm finds
an equivalence on the four events which recovers (an ST-structure
isomorphic to) $H$: in this case we equate $q_1\sim q_3$ and
$q_2\sim q_4$, %
which corresponds to the standard way of identifying events in HDA as
opposite sides of a filled-in square when it exists.

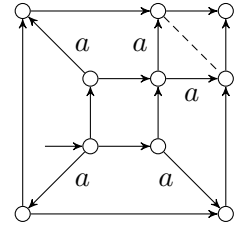
\begin{figure}[bp]
  \centering
  \begin{tikzpicture}[>=stealth', x=.9cm, y=.9cm]
    \node[state, initial] (9) at (2,1) {};
    \node[state] (5) at (2,2) {};
    \node[state] (10) at (3,1) {};
    \node[state] (6) at (3,2) {};
    \path (9) edge (5);
    \path (5) edge (6);
    \path (9) edge (10);
    \path (10) edge (6);

    \node[state] (1) at (1,3) {};
    \node[state] (2) at (3,3) {};
    \node[state] (3) at (4,3) {};
    \node[state] (7) at (4,2) {};
    \node[state] (11) at (1,0) {};
    \node[state] (12) at (4,0) {};

    \path (11) edge (1);
    \path (1) edge (2);
    \path (2) edge (3);
    \path (7) edge (3);
    \path (11) edge (12);
    \path (12) edge (7);

    \path (5) edge node[right] {\;$a$} (1);
    \path (6) edge node[left] {$a$} (2);
    \path (6) edge node[below] {$a$} (7);
    \path (9) edge node[right] {\;$a$} (11);
    \path (10) edge node[left] {$a$\;} (12);

    \path (2) edge[-, densely dashed] (7);
  \end{tikzpicture}
  \caption{%
    \label{fig:nosculpt}
    A one-dimensional acyclic HDA which cannot be sculpted.}
\end{figure}

Another example is shown in Figure~\ref{fig:nosculpt}.  This
one-dimensional acyclic HDA cannot be sculpted, and the algorithm
detects this by noting that (1) all the $a$-labeled transitions indeed
need to be the same event, but then (2) the two states connected
with a dashed line need to be identified, so that the ST-structure
covering cannot be isomorphic to the original HDA model.

In Section~\ref{se:euclid} we make the connection between the
combinatorial and geometric models and show that HDA can be sculpted
precisely if they are Euclidean.  This necessitates a few notions from
directed topology which can be found in appendix.

\begin{figure}[tp]
  \centering
  \begin{tikzpicture}[x=3.1cm, y=1.4cm]
    \node[rectangle] (rChu) at (-.2,0) {\normalsize reg.\ Chu spaces over
      $\three$};
    \node[rectangle] (Chu) at (1.2,0) {\normalsize Chu spaces over $\three$};
    \node[rectangle] (rST) at (-.1,-1) {\normalsize regular ST-structures};
    \node[rectangle] (ST) at (1.1,-1) {\normalsize ST-structures\vphantom p};
    \node[rectangle] (Euc) at (-1.1,-2) {\normalsize Euclidean complexes};
    \node[rectangle] (Scu) at (.2,-2) {\normalsize Sculptures};
    \node[rectangle] (aHDA) at (1.25,-2) {\normalsize acyclic HDA};
    \node[rectangle] (HDA) at (2,-2) {\normalsize HDA\vphantom p};

    \path[right hook-stealth'] (rChu) edge (Chu);
    \path[right hook-stealth'] (rST) edge (ST);
    \path[<->] (Euc) edge node[below] {Thm.~\ref{th:scu-euc}} (Scu);
    \path[right hook-stealth'] (Scu) edge node[below] {Thm.~\ref{th:conc}} (aHDA);
    \path[right hook-stealth'] (aHDA) edge (HDA);

    \path (rChu) -- (rST) coordinate[midway] (auxrcs);
    \path[<->] (rChu.south-|auxrcs) edge node [right]
    {\cite{Johansen16STstruct}} (rST.north-|auxrcs);

    \path (Chu) -- (ST) coordinate[midway] (auxcs);
    \path[<->] (Chu.south-|auxcs) edge node [left]
    {\cite{Johansen16STstruct}} (ST.north-|auxcs);

    \path[<->] (Scu|-rST.south) edge node[left] {Thm.~\ref{prop_stSculptst}} (Scu);

    \path (Chu.340) edge[out=-60, in=60, -, densely dashed]
    (aHDA.25);
  \end{tikzpicture}
  \caption{%
    \label{fig:contribs}
    Contributions of this paper.  (All inclusions are strict.)}
\end{figure}
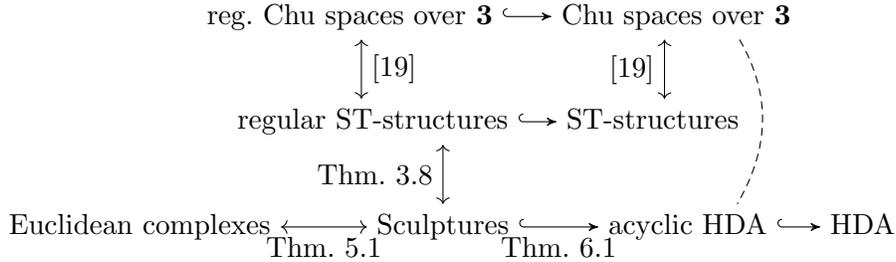

Figure~\ref{fig:contribs} sums up the relations between the different
models which we expose in this paper.  (The dashed line indicates the
common belief that Chu spaces over $\three$ and acyclic HDA are
equivalent, which we prove not to be the case.)

\section{HDA, ST-Structures, and Chu Spaces}\label{sec_background}

HDA are automata in which independence of events is indicated by
higher-dimensional structure. HDA consist of states, transitions,
and cubes of different dimensions which represent events running
concurrently.

\paragraph*{Precubical sets}

A \emph{precubical set} is a graded set $Q= \bigcup_{ n\in \Nat} Q_n$,
with $Q_n\cap Q_m= \emptyset$ for all $n\ne m$, together with mappings
$s_{ k, n}, t_{ k, n}:Q_n\to Q_{ n- 1}$, $k= 1,\dots, n$,
satisfying the following \emph{precubical identities},
for $\alpha, \beta\in\{ s, t\}$,
\begin{equation}
  \label{eq:precubical}
  \alpha_{ k, n- 1} \beta_{ \ell, n}= \beta_{ \ell- 1, n- 1}
  \alpha_{ k, n} \qquad( k< \ell)
\end{equation}

\begin{figure}[bp]
\vspace{-2ex}
  \centering
  \begin{tikzpicture}[>=stealth']
    \path[fill=black!10] (0,0) to (2,0) to (2,2) to (0,2) to (0,0);
    \node[state] (00) at (0,0) {};
    \node[state] (10) at (2,0) {};
    \node[state] (01) at (0,2) {};
    \node[state] (11) at (2,2) {};
    \path (00) edge (01);
    \path (00) edge (10);
    \path (01) edge (11);
    \path (10) edge (11);
    \node at (1,1) {$q$};
    \node at (-.4,1.05) {$s_1 q$};
    \node at (2.4,1.05) {$t_1 q$};
    \node at (1,.25) {$\vphantom{t}s_2 q$};
    \node at (1,1.75) {$t_2 q$};
    \node at (-.2,-.35) {$\vphantom{t} s_1 s_2 q= s_1 s_1 q$};
    \node at (-.2,2.35) {$s_1 t_2 q= t_1 s_1 q$};
    \node at (2.2,-.35) {$t_1 s_2 q= s_1 t_1 q$};
    \node at (2.2,2.35) {$t_1 t_2 q= t_1 t_1 q$};
  \end{tikzpicture}
  \caption{%
    \label{fi:2cubefaces-full}
    A $2$-cell $q$ with its four faces $s_1 q$, $t_1 q$, $s_2 q$,
    $t_2 q$ and four corners.}
\end{figure}
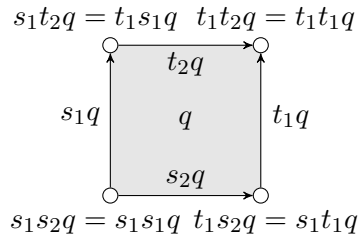
Elements of $Q_n$ are called \emph{$n$-cells} (or simply cells), and
for $q\in Q_n$, $n= \dim q$ is its \emph{dimension}.  The mappings
$s_{ k, n}$ and $t_{ k, n}$ are called \emph{face maps}, and we will
usually omit the extra subscript $n$ and simply write $s_k$ and $t_k$.
Intuitively, each $n$-cell $q\in Q_n$ has $n$ \emph{lower faces}
$s_1 q,\dotsc, s_n q$ and $n$ \emph{upper faces}
$t_1 q,\dotsc, t_n q$, whereas the precubical identity expresses the fact
that $( n- 1)$-faces of an $n$-cell meet in common $( n- 2)$-faces;
see Figure~\ref{fi:2cubefaces-full} for an example.

\textit{Morphisms} $f: Q\to R$ of precubical sets are graded functions
$f=\{ f_n: Q_n\to R_n\}_{ n\in \Nat}$ which commute with the face
maps: $\alpha_k\circ f_n= f_{ n- 1}\circ \alpha_k$ for all
$n\in \Nat$, $k\in\{ 1,\dots, n\}$, and $\alpha\in\{ s, t\}$.  This
defines a category $\pCub$ of precubical sets.
A precubical morphism 
is an \emph{embedding} if it is
injective; in that case we write $f: Q\hookrightarrow R$.  $Q$ and $R$
are isomorphic, denoted $Q\isomorphic R$, if there is a bijective
morphism $Q\to R$.

If two cells 
$q,q'\in Q$ 
are in a face relation 
$q= \alpha^1_{ i_1}\dotsm \alpha^n_{ i_n} q'$,
for
$\alpha^1,\dotsc, \alpha^n\in\{ s, t\}$, then this sequence can be
rewritten in a unique way, using the precubical
identities~\eqref{eq:precubical}, so that the indices
$i_1<\dotsm< i_n$, see~\cite{GrandisM02site}.  $Q$ is said to be
\emph{non-selflinked} if up to this rewriting, there is at most one
face relation between any of its cells, that is, it holds for all
$q, q'\in Q$ that there exists at most one index sequence
$i_1<\dotsc< i_n$ such that
$q= \alpha^1_{ i_1}\dotsm \alpha^n_{ i_n} q'$ for
$\alpha^1,\dotsc, \alpha^n\in\{ s, t\}$.

In other words, $Q$ is non-selflinked iff any $q\in Q$ is
\emph{embedded} in $Q$, hence iff all $q$'s iterated faces are
genuinely different.
This conveys a geometric intuition of regularity and is frequently
assumed~\cite{Fajstrup05, Fajstrup06}, also in algebraic topology
\mbox{\cite[Def.~IV.21.1]{Bredon93top}}.  It means that for all cells
in $Q$, each of their faces (and faces of faces etc.) are present in
$Q$ as distinct cells.

\paragraph*{Higher-dimensional automata}

A precubical set $Q$ is \emph{finite} if $Q$ is finite as a set.  This
means that $Q_n$ is finite for each $n\in \Nat$ and that $Q$ is
\emph{finite-dimensional}: there exists $N\in \Nat$ such that
$Q_n= \emptyset$ for all $n> N$ (equivalently, $\dim q\le N$ for all
$q\in Q$).  In that case, the smallest such $N$ is called the
dimension of $Q$ and denoted $\dim Q= \max\{ \dim q\mid q\in Q\}$.
A \emph{higher-dimensional automaton} (HDA) is a finite %
precubical set $Q$ with a designated initial cell $\initialCell\in Q_0$.
Morphisms $f: Q\to Q'$ of HDA are precubical morphisms that
fix the initial cell, \ie~have $f( \initialCell)= \initialCell'$.
We often call cells from $Q_0$ and $Q_1$ respectively \emph{states} and \emph{transitions}.

Note that we only deal with \emph{unlabeled} HDA here, \ie~HDA without
labellings on transitions and/or higher cells.  We are interested here in the
\emph{events}, not in their labeling.

A \emph{step} in an HDA, 
with $q_{n}\in Q_{n}$, $q_{n-1}\in Q_{n-1}$, and $1\leq i\leq n$,  
is either
\begin{equation*}
  q_{n-1}\transition{s_{i}}q_{n}\text{ with }s_{i}(q_{n})=q_{n-1}\text{\ \ \ or\ \ \ }
  q_{n}\transition{t_{i}}q_{n-1}\text{ with }t_{i}(q_{n})=q_{n-1}\,.
\end{equation*}
A \emph{path}
$\pi\defequal
q^{0}\transition{\alpha^{1}}q^{1}\transition{\alpha^{2}}q^{2}\transition{\alpha^{3}}\dots$
is a sequence of steps $q^{j-1}\transition{\alpha^{j}}q^{j}$, with
$\alpha^{j}\in\{s,t\}$.
The first cell is denoted $\startPath{\pi}$ and the ending cell in a
finite path is $\finishPath{\pi}$.
The string $\alpha^1\dots \alpha^n$ consisting of letters $s$, $t$ is \emph{the type} of the path $\pi$.

Two paths are \emph{elementary homotopic}~\cite{DBLP:journals/tcs/Glabbeek06}, %
denoted $\pi\homotopicHDA\pi'$, if one can be obtained
from the other by replacing, for $q \in Q$ and $i < j$, either 
(1) a segment
$\transition{s_i}q\transition{s_j}$ by
$\transition{s_{j-1}}q'\transition{s_i}$, 
(2) a segment
$\transition{t_j}q\transition{t_i}$ by
$\transition{t_i}q'\transition{t_{j-1}}$, 
(3) a segment
$\transition{s_i}q\transition{t_j}$ by
$\transition{t_{j-1}}q'\transition{s_i}$, or 
(4) a segment
$\transition{s_j}q\transition{t_i}$ by
$\transition{t_{i}}q'\transition{s_{j-1}}$.  \emph{Homotopy} is the
reflexive and transitive closure of the above and is denoted the
same. Two homotopic paths thus share their respective start and end
cells.

A cell $q'$ in an HDA $Q$ is \emph{reachable} from another cell $q$ if
there exists a path $\pi$ with $\startPath \pi= q$ and
$\finishPath \pi= q'$.  $Q$ is said to be \emph{connected} if any cell
is reachable from the initial state $\initialCell$.  $Q$ is \emph{acyclic}
if there are no two different cells $q, q'$ in $Q$ such that $q'$ is
reachable from $q$ and $q$ is reachable from $q'$.

If an HDA is not connected, then it contains cells which are not
reachable during any computation.
We will hence assume all HDA to be connected.

\paragraph*{Universal event labeling}
Let $Q$ be a precubical set and define $\eventEquivBulk$ to be the
equivalence relation on $Q_1$ spanned by
$\{( s_i q, t_i q)\mid q\in Q_2, i\in\{ 1, 2\}\}$.  Let
$\universalEvents{Q}=\quotientofwrt{Q_1}{\mathord{\eventEquivBulk}}$ be the set of equivalence
classes; this is called the set of \emph{universal labels} of $Q$. The universal label of a transition $q_1$ will be denoted by $\lambda(q_1)$.

For every precubical morphism $f:Q\to R$ and transitions $e,e'\in Q_1$, $e\eventEquivBulk e'$ implies $f(e)\eventEquivBulk f(e')$. As a consequence, $f$ induces a map between the sets of universal labels fitting into the diagram:
  \begin{equation*}
    \xymatrix{%
      Q
      	\ar[r]^-{f}
      	\ar[d]_{\lambda} &
      R
      	\ar[d]^{\lambda} \\
      \universalEvents{Q} \ar[r]^-{\universalEvents{f}} &
      \universalEvents{R}
    }
  \end{equation*}
This makes $\universalEvents{}$ a functor from the category of precubical sets $\pCub$ into the category of sets.

$Q$ is said to be \emph{inherently self-concurrent} if there is
$q\in Q_2$ for which $s_1 q\eventEquivBulk s_2 q$ or (equivalently)
$t_1 q\eventEquivBulk t_2 q$.
In that case, $\universalEvents{Q}$ does not identify
events, as there are cells in which more than one occurrence of an
event is active.  We say that $Q$ 
is \emph{consistent} 
if it is not
inherently self-concurrent.

\begin{exa}\label{exa_NOTconsistent}
The examples (\ref{exa_NOTconsistent_1}) and (\ref{exa_NOTconsistent_2}) are not consistent, though the first one is selflinked, whereas the second one is non-selflinked. Example (\ref{exa_NOTconsistent_3}) is consistent and selflinked.
\begin{enumerate}
\item\label{exa_NOTconsistent_1} Consider the HDA with three cells $\{Q_{0}=\{q_{0}\},Q_{1}=\{q_{1}\},Q_{2}=\{q_{2}\}\}$ where all the four maps of the square point to the same transition $\alpha_{i}(q_{2})=q_{1}$, for $\alpha\in\{s,t\}$ and $i\in\{1,2\}$ and the two maps of the transition point to the same state $\alpha_{1}(q_{1})=q_{0}$, which is also the initial state.
For visual help we draw states as circles, transitions as squares, and 2-cells as hexagons.
\tikzset{zerocell/.style={circle, draw, minimum size=0.6cm, inner sep=0pt}}
\tikzset{onecell/.style={draw, minimum size=0.6cm, inner sep=0pt}}
\tikzset{twocell/.style={draw,regular polygon, regular polygon sides=6, minimum size=0.7cm, inner sep=0pt}}
\begin{figure}[h]
\begin{tikzpicture}[>=stealth']
	\node[zerocell] (0) at (6,0) {$q_0$};
	\node[onecell] (1) at (3,0) {$q_1$};
	\node[twocell] (2) at (0,0) {$q_2$};
	\path[bend left=20] (2) edge (1);
	\path[bend left=60] (2) edge (1);
	\path[bend right=20] (2) edge (1);
	\path[bend right=60] (2) edge (1);
	\path[bend left=20] (1) edge (0);
	\path[bend right=20] (1) edge (0);
	\node at (1.5,0.55) {\scriptsize $s_1$};
	\node at (1.5,-0.55) {\scriptsize $t_1$};
	\node at (4.5,0.55) {\scriptsize $s_1$};
	\node at (4.5,-0.55) {\scriptsize $t_1$};
	\node at (1.5,1.15) {\scriptsize $s_2$};
	\node at (1.5,-1.15) {\scriptsize $t_2$};
\end{tikzpicture}
\end{figure}

\item\label{exa_NOTconsistent_2} Consider the HDA formed of three  squares $q_{2},q'_{2},q''_2\in Q_{2}$ that are adjacent, \ie~$t_{1}(q_{2})=s_{1}(q'_{2})$ and $t_{1}(q'_{2})=s_{1}(q''_{2})$, but the first square shares with the third one a transition, namely $s_{1}(q_{2})=s_{2}(q''_{2})=e$.
\begin{figure}[h]
\begin{tikzpicture}[>=stealth']
    \path[fill=black!10] (0,0) to (6,0) to (6,2) to (0,2) to (0,0);
    \node[state] (00) at (0,0) {};
    \node[state] (10) at (2,0) {};
    \node[state] (01) at (0,2) {};
    \node[state] (11) at (2,2) {};
    \node[state] (20) at (4,0) {};
    \node[state] (21) at (4,2) {};
    \node[state] (30) at (6,0) {};
    \node[state] (31) at (6,2) {};
    \path [very thick] (00) edge (01);
    \path (00) edge (10);
    \path (01) edge (11);
    \path (10) edge (11);
    \path (10) edge (20);
    \path (11) edge (21);
    \path (20) edge (21);
    \path [very thick] (20) edge (30);
    \path (21) edge (31);
    \path (30) edge (31);
    \node at (1,1) {$q_2$};
    \node at (3,1) {$q'_2$};
    \node at (5,1) {$q''_2$};
    \node [below] at (5,0) {$e$};
    \node [left] at (0,1) {$e$};
  \end{tikzpicture}
\end{figure}

\item\label{exa_NOTconsistent_3} Consider the HDA formed of a single 2-cell $q_2$ having two of its vertices identified, \ie~$t_1(s_1(q_2))=v=t_1(s_2(q_2))$.
\begin{figure}[h]
\begin{tikzpicture}[>=stealth']
    \path[fill=black!10] (0,0) to (2,0) to (2,2) to (0,2) to (0,0);
    \node[state] (00) at (0,0) {};
    \node[state] (10) at (2,0) {};
    \node[state] (01) at (0,2) {};
    \node[state] (11) at (2,2) {};
    \path (00) edge (10);
    \path (00) edge (01);
    \path (01) edge (11);
    \path (10) edge (11);
    \node [below right] at (2,0) {$v$};
    \node [above left] at (0,2) {$v$};
    \node at (1,1) {$q_2$};

    \node[twocell] (q2) at (6,1) {$q_2$};
    \node[onecell] (a) at (4,1) {$a$};
    \node[onecell] (b) at (8,1) {$b$};
    \node[onecell] (d) at (6,3) {$d$};
    \node[onecell] (c) at (6,-1) {$c$};
    \node[zerocell] (x) at (4,-1) {$x$};
    \node[zerocell] (y) at (8,3) {$y$};
    \node[zerocell] (v) at (4,3) {$v$};

	\path (q2) edge (a) edge (b) edge (c) edge (d);
	\path (a) edge (x) edge (v);    
	\path (c) edge (x) edge (v);    
	\path (b) edge (y) edge (v);    
	\path (d) edge (y) edge (v);    

	\node[above=-2pt] at (5.5,1) {\scriptsize $s_1$};
	\node[above=-2pt] at (6.5,1) {\scriptsize $t_1$};
	\node[above=-2pt] at (5.5,3) {\scriptsize $s_1$};
	\node[above=-2pt] at (6.5,3) {\scriptsize $t_1$};
	\node[below=-2pt] at (5.5,-1) {\scriptsize $s_1$};
	\node at (5.4,-0.2) {\scriptsize $t_1$};
	\node[left=-2pt] at (4,1.8) {\scriptsize $t_1$};
	\node[left=-2pt] at (4,0.2) {\scriptsize $s_1$};
	\node[right=-2pt] at (6,1.6) {\scriptsize $t_2$};
	\node[right=-2pt] at (6,0.4) {\scriptsize $s_2$};
	\node[right=-2pt] at (8,1.8) {\scriptsize $t_1$};
	\node at (7.2,1.6) {\scriptsize $s_1$};
\end{tikzpicture}
\end{figure}
\end{enumerate}
\end{exa}

Let ${\eventless} \subseteq \universalEvents{Q}\times\universalEvents{Q}$ be the transitive closure of the relation 
\[
\{(\lambda(s_2(q)),\lambda(s_1(q)))\mid q\in Q_2\}.
\]
We say that $Q$ is 
\emph{ordered} 
if for every $a,b\in\universalEvents{Q}$ the conditions $a\eventless b$ and $b\eventless a$ cannot hold simultaneously or, equivalently, $\eventless$ is antisymmetric and $a\not\eventless a$ for all $a\in\universalEvents{Q}$.
For a precubical morphism $f:Q\to R$, the induced map $\universalEvents{f}:\universalEvents{Q}\to \universalEvents{R}$ preserves the relation $\eventless$. This makes $\universalEvents{}$ a functor into the category of sets with a transitive relation and relation-preserving maps.

If $Q$ is not consistent, then it is not ordered. 
Indeed, if $a=\lambda(s_1(q))= \lambda(s_2(q))$ for some $q\in Q_2$, then $a\eventless a$, which  
excludes ordered.

For every square $q$ we can assign a pair of universal labels $(\lambda(\alpha_2(q)),\lambda(\beta_1(q)))$, which does not depend on the choice $\alpha,\beta\in \{s,t\}$. This generalizes for higher-dimensional cubes: for $q\in Q_n$ and $i\in \{1,\dotsc,n\}$ let 
\[
	\lambda_i(q)=\lambda(s_1(s_2(\dots(s_{i-1}(s_{i+1}(\dots(s_n(q))\dots)). 
\]
Again, we can replace some of $s$'s with $t$'s and get the same result. Denote $\lambda(q)=(\lambda_1(q),\dotsc,\lambda_n(q))$.

\begin{lem}\label{le:multilabels}
	Some properties:
	\begin{enumerate}
	\item
	For $q\in Q_n$, $\alpha\in\{s,t\}$,
	\[
		\lambda(\alpha_i(q))=(\lambda_1(q),\dotsc,\lambda_{i-1}(q),\lambda_{i+1}(q),\dotsc,\lambda_n(q)).
	\]
	\item
	$\lambda_1(q)\eventless \lambda_2(q)\eventless \dots \eventless \lambda_n(q)$.
	\item
	$Q$ 
is consistent
if{}f for every $n$ and $q\in Q_n$, all $\lambda_i(q)$ are different.
	\item
	If $Q$ is 
ordered, 
then $\lambda_i(q)\eventless \lambda_j(q)$ implies $i< j$.
	\end{enumerate}
\end{lem}
\begin{proof}
\begin{enumerate}
\item
	For $j\geq i$ we have
	\begin{multline*}
		\lambda_{i}(\alpha_j(q))=
		\lambda(s_1 s_2\dots s_{i-1} s_{i+1}\dots s_{n-2} s_{n-1}\alpha_j(q))
		=\lambda(s_1 s_2\dots s_{i-1} s_{i+1}\dots s_{n-2} \alpha_j s_{n}(q))\\
		=\lambda(s_1 s_2\dots s_{i-1} s_{i+1}\dots \alpha_j s_{n-1}  s_{n}(q))
		=\dots 
		=\lambda(s_1 s_2\dots s_{i-1} s_{i+1}\dots s_{j-1}\alpha_js_{j+1}\dots s_{n-1}  s_{n}(q))\\
		=\lambda(s_1 s_2\dots s_{i-1} s_{i+1}\dots s_{j-1} s_j s_{j+1}\dots s_{n-1}  s_{n}(q))=\lambda_i(q)
	\end{multline*}
	For $j<i$,
	\begin{multline*}
		\lambda_{i}(\alpha_j(q))=
		\lambda(s_1 s_2\dots s_{i-1} s_{i+1}\dots s_{n-2} s_{n-1}\alpha_j(q))
		=\dots
		=\lambda(s_1 s_2\dots s_{i-1}\alpha_j s_{i+2}\dots s_{n-1} s_{n}(q))\\
		=\lambda(s_1 s_2\dots \alpha_j s_{i} s_{i+2}\dots s_{n-1} s_{n}(q))
		=\dots 
		=\lambda(s_1 s_2\dots s_{j-1}\alpha_j s_{j+1}\dots s_{i} s_{i+2}\dots s_{n-1}  s_{n}(q))\\
		=\lambda(s_1 s_2\dots s_{i} s_{i+2}\dots s_{n-1}  s_{n}(q))=\lambda_{i+1}(q).
	\end{multline*}
	\item
	Fix $q\in Q_n$ and $i\in \{1,\dots,n-1\}$ and let $z=s_1s_2\dots s_{i-1}s_{i+2}\dots s_{n}(q)$. 
	Using (1) we obtain that $\lambda_1(z)=\lambda_i(q)$, $\lambda_2(z)=\lambda_{i+1}(q)$; therefore, $\lambda_i(q)\eventless\lambda_{i+1}(q)$.
	\item
	If $Q$ is not consistent, then there exists $q\in Q_2$ such that
	\[
		\lambda_1(q)=\lambda(s_2(q))=\lambda(s_1(q))=\lambda_2(q).\]
	If $\lambda_i(q)=\lambda_j(q)$ for some $q\in Q_n$, $i<j$, then
	\begin{multline*}
	\lambda(s_2(s_1\dots s_{i-1}s_{i+1}\dots s_{j-1}s_{j+1}\dots s_n(q)))
	=\lambda(s_1 \dots s_{j-1}s_{j+1}\dots s_n(q))
		=\lambda_j(q)\\
		=\lambda_i(q)
		=\lambda(s_1\dots s_{i-1}s_{i+1}\dots  s_n(q))=\lambda(s_1(s_1\dots s_{i-1}s_{i+1}\dots s_{j-1}s_{j+1}\dots s_n(q))).
	\end{multline*}
	\item
		If $\lambda_i(q)\eventless \lambda_j(q)$ for $j\leq i$, then either $\lambda_j(q)\eventless \lambda_i(q)$ (if $j<i$) or $\lambda_i(q)\eventless \lambda_i(q)$ (if $i=j$). In both cases $Q$ cannot be ordered.\qedhere
		\end{enumerate}
\end{proof}

\begin{exa}\label{exa_NOTordered}
The following HDA is not ordered (due to a ``wrong'' numbering of the face maps), but we can swap the directions of one of the squares to obtain an ordered HDA.
Consider the following HDA $Q$:
\begin{itemize}
 \item  $Q_2 = \{A, B, C\}$,
 \item  $Q_1 = \{a, b, c, d, e, f, g, h, i\}$
 \item  $Q_0 = \{0, 1, 2, 3, 4, 5, 6, 7\}$, $0$ being the initial state,
\end{itemize}
with the following face maps:
\begin{itemize}
 \item  $s_2 (A) = s_1 (B) = a$,
 \item  $s_2 (B) = s_1 (C) = b$,
 \item  $s_2 (C) = s_1 (A) = c$,
 \item  $s_1 (a) = s_1 (b) = s_1 (c) = 0$
 \item  the other face maps are not of importance for the example.
\end{itemize}

\tikzset{zerocell/.style={circle, draw, minimum size=0.6cm, inner sep=0pt}}
\tikzset{onecell/.style={draw, minimum size=0.6cm, inner sep=0pt}}
\tikzset{twocell/.style={draw,regular polygon, regular polygon sides=6, minimum size=0.7cm, inner sep=0pt}}
\begin{figure}[h]
\begin{tikzpicture}[>=stealth',every node/.style={minimum size=0.8cm}]
	\node[zerocell] (0) at (0,0) {$0$};
	\node[zerocell] (1) at (-4,0) {$1$};
	\node[zerocell] (2) at (-2,3.32) {$2$};
	\node[zerocell] (3) at (2,3.32) {$3$};
	\node[zerocell] (4) at (4,0) {$4$};
	\node[zerocell] (5) at (2,-3.32) {$5$};
	\node[zerocell] (6) at (-2,-3.32) {$6$};
	\node[onecell] (a) at (-2,0) {$a$};
	\node[onecell] (c) at (1,1.66) {$c$};
	\node[onecell] (b) at (1,-1.66) {$b$};
	\node[onecell] (d) at (-3,1.66) {$d$};
	\node[onecell] (e) at (0,3.32) {$e$};
	\node[onecell] (f) at (3,1.66) {$f$};
	\node[onecell] (g) at (3,-1.66) {$g$};
	\node[onecell] (h) at (0,-3.32) {$h$};
	\node[onecell] (i) at (-3,-1.66) {$i$};
	\node[twocell] (A) at (-1, 1.66) {$A$};
	\node[twocell] (B) at (-1, -1.66) {$B$};
	\node[twocell] (C) at (2, 0) {$C$};
	\path (a) edge (0) edge (1);
	\path (b) edge (0) edge (5);
	\path (c) edge (0) edge (3);
	\path (d) edge (1) edge (2);
	\path (e) edge (2) edge (3);
	\path (f) edge (3) edge (4);
	\path (g) edge (4) edge (5);
	\path (h) edge (5) edge (6);
	\path (i) edge (6) edge (1);
	\path (A) edge (a) edge (c) edge (d) edge (e);
	\path (B) edge (a) edge (b) edge (h) edge (i);
	\path (C) edge (b) edge (c) edge (f) edge (g);
	\node[above=-7pt] at (-1,0) {\scriptsize $s_1$};
	\node[above=-7pt] at (0,1.66) {\scriptsize $s_1$};
	\node[above=-7pt] at (1,3.32) {\scriptsize $s_1$};
	\node[above=-7pt] at (0,-1.66) {\scriptsize $s_2$};
	\node[above=-7pt] at (1,-3.32) {\scriptsize $s_1$};
	\node[above=-7pt] at (-3,0) {\scriptsize $t_1$};
	\node[above=-7pt] at (-2,1.66) {\scriptsize $t_1$};
	\node[above=-7pt] at (-1,3.32) {\scriptsize $t_1$};
	\node[above=-7pt] at (-2,-1.66) {\scriptsize $t_2$};
	\node[above=-7pt] at (-1,-3.32) {\scriptsize $t_1$};

	\node[right=-5pt] at (0.5,0.83) {\scriptsize $s_1$};
	\node[right=-5pt] at (-1.5,0.83) {\scriptsize $s_2$};
	\node[right=-5pt] at (-3.5,0.83) {\scriptsize $s_1$};
	\node[right=-5pt] at (1.5,-0.83) {\scriptsize $s_1$};
	\node[right=-5pt] at (2.5,-2.49) {\scriptsize $s_1$};
	\node[right=-5pt] at (1.5,2.49) {\scriptsize $t_1$};
	\node[right=-5pt] at (-0.5,2.49) {\scriptsize $t_2$};
	\node[right=-5pt] at (-2.5,2.49) {\scriptsize $t_1$};
	\node[right=-5pt] at (2.5,0.83) {\scriptsize $t_1$};
	\node[right=-5pt] at (3.5,-0.83) {\scriptsize $t_1$};

	\node[left=-5pt] at (0.5,-0.83) {\scriptsize $s_1$};
	\node[left=-5pt] at (-1.5,-0.83) {\scriptsize $s_1$};
	\node[left=-5pt] at (-3.5,-0.83) {\scriptsize $s_1$};
	\node[left=-5pt] at (1.5,0.83) {\scriptsize $s_2$};
	\node[left=-5pt] at (2.5,2.49) {\scriptsize $s_1$};
	\node[left=-5pt] at (1.5,-2.49) {\scriptsize $t_1$};
	\node[left=-5pt] at (-0.5,-2.49) {\scriptsize $t_1$};
	\node[left=-5pt] at (-2.5,-2.49) {\scriptsize $t_1$};
	\node[left=-5pt] at (2.5,-0.83) {\scriptsize $t_2$};
	\node[left=-5pt] at (3.5,0.83) {\scriptsize $t_1$};

    \path[fill=black!10] (0-10,0) to (1-10,1.66) to (3-10,1.66) to (4-10,0) to (3-10,-1.66) to (1-10,-1.66);
    \node[state] (q0) at (2-10,0) {};
    \node[state] (q1) at (0-10,0) {};
    \node[state] (q2) at (1-10,1.66) {};
    \node[state] (q3) at (3-10,1.66) {};
    \node[state] (q4) at (4-10,0) {};
    \node[state] (q5) at (3-10,-1.66) {};
    \node[state] (q6) at (1-10,-1.66) {};
    \path (q0) edge (q1) edge (q3) edge (q5);
    \path (q1) edge (q2) edge (q6);
    \path (q3) edge (q2) edge (q4);
    \path (q5) edge (q6) edge (q4);
    \node [above=-7pt] at (1-10,0) {$a$};
    \node [right=-7pt] at (2.5-10,0.83) {$c$};
    \node [left=-7pt] at (2.5-10,-0.83) {$b$};
    \node at (1.5-10,0.83) {$A$};
    \node at (1.5-10,-0.83) {$B$};
    \node at (3-10,0) {$C$};

\end{tikzpicture}
\end{figure}

This is not ordered because antisymmetry is broken by the following: 
$\lambda(a)=\lambda_1(A)\eventless\lambda_2(A)=\lambda(c)$ and 
$\lambda(b)=\lambda_1(B)\eventless\lambda_2(B)=\lambda(a)$ and 
$\lambda(c)=\lambda_1(C)\eventless\lambda_2(C)=\lambda(b)$.
Note that $Q$ is almost the union of the three start faces of a 3-cell, except that the face maps of $A$ are in the ``wrong'' order, \ie~in a cube we would have the following:
\begin{itemize}
 \item  $s_1 (A) = s_1 (B) = a$,
 \item  $s_2 (B) = s_1 (C) = b$,
 \item  $s_2 (C) = s_2 (A) = c$.
\end{itemize}
However, we can create a slightly different HDA (called symmetric variant) that would be ordered, by only changing the maps of one of the offending squares, \ie~take the HDA from above only with $A'$ instead of $A$, such that $s_1( A')= a$ and $s_2( A')= c$. We could have, alternatively, reordered the maps of $B$ or $C$ to obtain other ordered symmetric variants.
\end{exa}

In the rest of the paper we consider only ordered HDAs.  As we show
later, in Proposition~\ref{prop_orderedNeededForSculpting}, only
ordered HDAs can be sculpted.  The result below shows that
this is not a restriction, as any consistent precubical set can be
ordered by re-arranging its face maps.

\begin{prop}
For every consistent precubical set $Q$ there exists 
an ordered
precubical set $Q'$ that is a symmetric variant of $Q$.
\end{prop}

By a ``symmetric variant'' we mean that $Q$ and $Q'$ are isomorphic when regarded as \emph{symmetric} precubical sets \cite{GrandisM02site}. In particular, there is a bijection between the set of paths on $Q$ and the set of paths on $Q'$ (see the proof details in Appendix~\ref{sec_ordered_pcs}).

\paragraph*{ST-structures}\label{subsec_STstructures}

An \emph{ST-configuration} over a finite set $E$ of \emph{events} is a
pair $( S, T)$ of sets $T\subseteq S\subseteq E$.  An
\emph{ST-structure} is a pair $\ST=( E, \sts)$ consisting of a finite
set $E$ of events and a set $\sts$ of ST-configurations over
$E$.

Intuitively, in an ST-configuration $( S, T)$ the set $S$ contains
events which have \emph{started} and $T$ contains events which have
\emph{terminated}.  Hence the condition $T\subseteq S$: only events
which have already started can terminate.  The events in
$S\setminus T$ are running \emph{concurrently}, and we call
$\num{ S\setminus T}$ the \emph{concurrency degree} of $( S, T)$.

The notion of having events which are currently running, \ie~started
but not terminated, is a key aspect captured by ST-structures and also
by HDA through their higher dimensional cells.  Other event-based
formalisms such as configuration structures~\cite{GlabbeekP95config,
  DBLP:journals/tcs/GlabbeekP09} or event
structures~\cite{DBLP:journals/tcs/NielsenPW81, Winskel86} cannot
express this.

A \emph{step} between two ST-configurations is either
\begin{description}
\item[s-step] $(S,T)\transitions{e}(S',T')$ with $T= T'$, $e\notin S$
  and $S'= S\cup\{ e\}$, or
\item[t-step] $(S,T)\transitiont{e}(S',T')$ with $S= S'$, $e\notin T$,
  and $T'= T\cup\{ e\}$.
\end{description}
When the type is unimportant we write $\smash{\transition{e}}$.
A \emph{path} of an ST-structure, denoted $\pi$, is a sequence of steps, where the end of one is the beginning of the next, \ie
\[
\pi\defequal(S,T)\transition{e}(S',T')\transition{e'}(S'',T'')\dots
\]
A path is \emph{rooted} if it starts in $(\emptyset,\emptyset)$.
An ST-structure $\ST=( E, \sts)$ is said to be
\begin{enumerate}
  \renewcommand{\labelenumi}{\theenumi}
  \renewcommand{\theenumi}{(\Alph{enumi})}%
\item \emph{rooted} if $( \emptyset, \emptyset)\in \sts$;
\item \emph{connected} if for any $( S, T)\in \sts$ there exists a
  rooted path ending in $(S,T)$;
\item \emph{closed under single events} if, for all $( S, T)\in \sts$
  and all $e\in S\setminus T$, also $( S, T\cup\{ e\})\in \sts$ and
  $( S\setminus\{ e\}, T)\in \sts$.
\end{enumerate}
$\ST$ is \emph{regular} if it satisfies all three conditions above.
Figure~\ref{fig:conflict}(center) shows a regular ST-structure.

ST-structures were introduced in~\cite{Johansen16STstruct} as an
event-based counterpart of HDA that are also a natural extension of
configuration structures and event structures.

The notions of rootedness and connectedness for ST-structures are
similar to connectedness for HDA.  The notion of being closed under
single events mirrors the fact that cells in HDA have all their faces,
and (by non-selflinkedness) these are all distinct.
Thus regularity is assumed in some of the results below.

A \emph{morphism} of ST-structures $( E, \sts)\to( E', \sts')$ is a
partial function $f: E\parto E'$ of events which preserves
ST-configurations (\ie~for all $( S, T)\in \sts$ we have $f( S, T):=( f( S), f( T))\in \sts'$)
and is \emph{locally total and injective} (\ie~for all $( S, T)\in \sts$, the restriction
$f\rest S: S\to E'$ is a total and injective function).
This defines a
category $\categoryST$ of ST-structures.
Two ST-structures are isomorphic, denoted $\ST \isomorphic \ST'$,
if there exists a bijective morphism between them.

\begin{defi}
  Let\/ $\ST=( E, \sts)$ be an ST-structure and
  $\mathord\sim\subseteq E\times E$ an equivalence relation.  The
  \emph{quotient} of\/ $\ST$ under $\sim$ is the ST-structure\/
  $\quotientofwrt{\ST}{\sim}=( \quotientofwrt{E}{\sim},
  \quotientofwrt{\sts}{\sim})$, with\/
  $\quotientofwrt{\sts}{\sim}=\{( \quotientofwrt{S}{\sim},
  \quotientofwrt{T}{\sim})\mid( S, T)\in \sts\}$.
\end{defi}

It is clear that $\quotientofwrt{\ST}{\sim}$ is again an ST-structure.
To ease notation we will sometimes denote
$( \quotientofwrt{S}{\sim}, \quotientofwrt{T}{\sim})= \quotientofwrt{(
  S, T)}{\sim}$.  The quotient map
$\quotientMap: \ST\to \quotientofwrt{\ST}{\sim}: e\mapsto[ e]_\sim$ is
generally not an ST-morphism, failing local injectivity.

\begin{defi}
  An equivalence relation $\mathord\sim\subseteq E\times E$ on an
  ST-structure $\ST=( E, \sts)$ is \emph{collapsing} if there is
  $( S, T)\in \sts$ and $e, e'\in S$ with $e\ne e'$ and $e\sim e'$.
  Otherwise, $\sim$ is \emph{non-collapsing}.
\end{defi}

\begin{lem}
  \label{le:quotientST}
  $\mathord\sim\subseteq E\times E$ is non-collapsing iff the quotient
  map
  $\quotientMap: \ST\to \quotientofwrt{\ST}{\sim}: e\mapsto[ e]_\sim$
  is an ST-morphism.
\end{lem}

\begin{proof}%
  If $\sim$ is collapsing, then $\quotientMap$ is not locally
  injective.  For the other direction, assume that $\quotientMap$ is
  not locally injective, then there is $( S, T)\in \sts$ and
  $e, e'\in S$ with $e\ne e'$ and
  $\quotientMap( e)= \quotientMap( e')$, thus $e\sim e'$: ergo $\sim$
  is collapsing.
\end{proof}

We want to compare sculptures and ST-structures. To this end, we
introduce a small but important modification to ST-structures which
consists in \emph{ordering} the events.

Let $E$ be a totally ordered set of events with ordering $<$.  An
\emph{ordered ST-structure} is an ST-structure on $E$, and morphisms
$f:( E, \sts)\to( E', \sts')$ of ordered ST-structures respect the
ordering, \ie~if $f( e_1)$ and $f( e_2)$ are defined and $e_1< e_2$, then
$f( e_1)< f( e_2)$.  This defines a category of ordered ST-structures.

\paragraph*{Chu spaces}
\label{subsec_Chu_andST}

The model of Chu spaces has been developed by Gupta and
Pratt~\cite{gupta94phd_chu, pratt95chu} in order to study the
event-state duality~\cite{Pratt92concur}.  A \textit{Chu space} over a
finite set $K$ is a triple $\chuPrat=(E,r,X)$ with $E$ and $X$ sets
and $r:E\times X\rightarrow K$ a function called the \emph{matrix} of
the Chu space.

Chu spaces can be viewed in various equivalent
ways~\cite[Chap.~5]{gupta94phd_chu}.  For our setting, we take the
view of $E$ as the set of events and $X$ as the set of configurations.
The structure $K$ is representing the possible values the events may
take, e.g.: $K= \boldsymbol2= \{0,1\}$ is the classical case of an event
being either not started ($0$) or terminated ($1$), hence Chu spaces
over $\boldsymbol2$ correspond to configuration
structures~\cite{GlabbeekP95config, DBLP:journals/tcs/GlabbeekP09} where an
order of $0<1$ is used to define the steps in the system, \ie~steps
between states must respect the increasing order when lifted pointwise
from $K$ to $X$.

ST-structures capture the ``during'' aspect in the event-based
setting, extending configuration structures with this
notion. Therefore we need another structure
$K= \three= \{0, \executing , 1\}$ with the order
$0 < \executing < 1$, introducing the value \executing\ to stand for
\textit{during}, or \textit{in transition}.  Note
that~\cite{gupta94phd_chu} studies Chu spaces over $\boldsymbol2$,
whereas Pratt proposed to study Chu spaces over $\three$ and other
structures in~\cite{Pratt03trans_cancel}.

A Chu space is \emph{extensional}~\cite{gupta94phd_chu} if it holds
for every $x\ne x'$ that there exists $e\in E$ such that
$r( e, x)\ne r( e, x')$.  We assume extensionality.  Using currying,
we can view a Chu space $( E, r, X)$ over $K$ as a structure
$X\subseteq K^{E}$ (this needs extensionality).  Consequently, we will
often write $x( e)$ instead of $r( e, x)$ below.
A Chu space is \emph{separable} \cite{Pratt02duality} (called T0 in
\cite{gupta94phd_chu}) if no two events are the same, that is, for all
$e\ne e'$ there is $x\in X$ such that $r(e,x)\ne r(e',x)$.

\begin{defi}[translations between ST and Chu]\label{def_STtoChu}
  For an ST-structure $\ST=( E, \sts)$ construct $\chu{E}{X}^{\ST}$
  the associated Chu space over $\three$ with $E$ the set of events
  from $\ST$, and $X\subseteq \three^{E}$ containing for each
  ST-configuration $(S,T)\in \sts$ the state $x^{(S,T)}\in X$ formed
  by assigning to each $e\in E$:
  \begin{itemize}
  \item $e \rightarrow 0$ if $e\not\in S$ and $e\not\in T$;
  \item $e \rightarrow \executing$ if $e\in S$ and $e\not\in T$;
  \item $e \rightarrow 1$ if $e\in S$ and $e\in T$.\footnote{The
      case  $e\notin S\wedge e\in T$ is dismissed by the
      requirement $T\subseteq S$ of ST-configurations.}
  \end{itemize}
  Call this mapping $\stintochu(S,T)$ when applied to an
  ST-configuration and $\stintochu(\ST)$ when applied to an
  ST-structure.
  The other way, we translate an extensional Chu space $(E, X)$
  into an ST-structure over $E$ with one ST-configuration $(S,T)^{x}$
  for each state $x\in X$ using the inverse of the above mapping.  We
  use $\chuintost(x)$ for the ST-configuration obtained from the event
  listing $x$.
\end{defi}

\begin{thmC}[{\cite[Sec.~3.4]{Johansen16STstruct}}]
  \label{prop_STstructChu3}
  For any ST-structure $\ST$,
  $\chuintost(\stintochu(\ST)) \isomorphic \ST$.  For any
  (extensional) Chu space $\chuPrat$ over $\three$,
  $\stintochu( \chuintost( \chuPrat)\isomorphic \chuPrat$.
\end{thmC}

\section{Sculptures}\label{sec_sculptures}

Inspired by the Chu notation for states, we define a \emph{bulk} in two equivalent ways, both of which can be seen as the complete $\chuPrat$ over $\three$.

\begin{defi}\label{def_canonical}
  Let $d\in \Nat$.  The $d$-dimensional bulk $\bulk d$ is the
  precubical set defined as follows.  For $n= 0,\dotsc, d$, let
  \begin{equation*}
    \bulk{d}_n=\big\{( x_1,\dotsc, x_d)\in \{ 0, \executing, 1\}^d\bigmid \num{\{
      i\mid x_i= \executing\}}= n\big\}
  \end{equation*}
  be the set of tuples with precisely $n$ occurrences of $\executing$.
  For $n= 1,\dotsc, d$ and $k= 1,\dotsc, n$, define face maps
  $s_k, t_k: \bulk{d}_n\to \bulk{d}_{ n- 1}$ as follows: for
  $x=( x_1,\dotsc, x_d)\in \bulk{d}_n$ with
  $x_{ i_1}=\dotsm= x_{ i_n}= \executing$, let
  $s_k x=( x_1,\dotsc, 0_{ i_k},\dotsc, x_d)$ and
  $t_k x=( x_1,\dotsc, 1_{ i_k},\dotsc, x_d)$ be the tuples with the
  $k$-th occurrence of $\executing$ set to $0$ or $1$, respectively.
\end{defi}

The \emph{initial state} $\initbulk{d}$ of the bulk $\bulk{d}$ is the
cell $( 0,\dotsc, 0)$.  This turns bulks into HDA.

Let $\bulkST{d}=( \evlist E, \sts)$ be the \emph{complete} ordered
ST-structure on $\evlist E=( 1,\dotsc, d)$, with
$\sts=\{( S, T)\mid T\subseteq S\subseteq \evlist E\}$.
There is a bijection between $\bulkST{d}$ and the bulk $\bulk d$ which maps a
configuration $( S, T)$ to the cell $( x_1,\dotsc, x_d)$ given by
\begin{equation*}
  x_i=
  \begin{cases}
    0 &\text{if } i\notin S\,, \\
    \executing &\text{if } i\in S\setminus T\,, \\
    1 &\text{if } i\in T\,,
  \end{cases}
\end{equation*}
\cf~Def.~\ref{def_STtoChu}, and using the inverse of the above when mapping the other way.  
This bijection induces face maps in \bulkST{d} as follows:
for
  $x=( x_1,\dotsc, x_d)\in \bulk{d}_n$ with
  $x_{ i_1}=\dotsm= x_{ i_n}= \executing$ and 
  $s_k x=( x_1,\dotsc, 0_{ i_k},\dotsc, x_d)$, i.e., with the
  $k$-th occurrence of $\executing$ set to $0$, then for the respective ST-configuration $(S,T)^{x}$ the map is $s_k((S,T)^{x})=(S,T)^{( x_1,\dotsc, 0_{ i_k},\dotsc, x_d)}=(S^{x}\setminus \{i_k\} , T^{x})$.
Conversely, \bulkST{d} can be equipped with face maps as $s_k((S,T))=(S\setminus \{i_k\},T)$ and $t_k((S,T))=(S,T\cup\{i_k\})$ 
with $i_k\in \evlist E$ being the $k^{th}$ event in the subset listing $\evlist E \restrictedToSet{S\setminus T}$.
We will use these two notions of bulk interchangeably and denote this bijection as $\bulkST{d} \bijection \bulk d$.

Let $d\le d'$ and $b:\{ 1,\dotsc, d\}\to\{ 1,\dotsc, d'\}$ a strictly
increasing function.  This defines an embedding, also denoted $b:
\bulk d\hookrightarrow \bulk{ d'}$, mapping any cell $( t_1,\dotsc,
t_d)$ to $( u_1,\dotsc, u_{ d'})$ given by
\begin{equation*}
  u_i=
  \begin{cases}
    t_j &\text{if } i= b( j)\,, \\
    0 &\text{if } i\notin \im( b)\,.
  \end{cases}
\end{equation*}
Every HDA morphism $\bulk d\to \bulk{ d'}$ is of this form (but not
every precubical morphism because these do not need to preserve the initial state), and there are no morphisms $\bulk d\to
\bulk{ d'}$ for $d> d'$.

\begin{lem}\label{lemma_universalEvents_in_Bulk}
Fix a bulk $\bulk{d}$ and take two transitions $(t_{1},\dots,t_{d})$ and $(u_{1},\dots,u_{d})$, and let $k$ be the unique index s.t.\ $t_{k}=\executing$, and the same for $u_{l}=\executing$. Then the two transitions represent the same event, \ie~$(t_{1},\dots,t_{d}) \eventEquivBulk (u_{1},\dots,u_{d})$, iff $k=l$.
Therefore, the set of universal events is $\universalEvents{\bulk{d}}=\{1,\dots,d\}$.
Moreover, the order $\eventless$ on $\universalEvents{\bulk{d}}$ agrees with the natural order on $\{1,\dots,d\}$.
\end{lem}

\begin{proof}
	For a transition $t=(t_1,\dots,t_d)$ in $\bulk{d}$ let $dir(t)$ be the unique index such that $t_{dir(t)}=\executing$. Let $y=(y_1,\dots,y_d)$ be an arbitrary 2-cube and its two unique indices $k<l$ such that $y_k=y_l=\executing$. We have $dir(s_1(y))=dir(t_1(y))=l$ and $dir(s_2(y))=dir(t_2(y))=k$. Therefore, $t\eventEquivBulk u$ implies $dir(t)=dir(u)$.
	
	We will show that $t\eventEquivBulk u$ for any two transitions $t,u$ with $dir(t)=dir(u)=k$. Induction with respect to the number $m$ of indices $i$ such that $t_i\neq u_i$. If $m=0$, then $t=u$ and there is nothing to prove. If $m>0$, then choose any index $j$ such that $t_j\neq u_j$; the square
	\[
		(t_1,\dots, t_{j-1}, \executing, t_{j+1},\dots, t_d)
	\]
	assures that $t\eventEquivBulk (t_1,\dots, t_{j-1}, u_j, t_{j+1},\dots, t_d)=t'$, and $t'\eventEquivBulk u$ by the inductive hypothesis. As a consequence, $dir: \universalEvents{\bulk{d}}\to \{1,\dots,d\}$ is a bijection.
	
	For a square $y$ with $y_k=y_l=\executing$,  $k<l$, we have $dir(s_2(y))=k<l=dir(s_1(y))$ and $\lambda(s_2(y)) \eventless \lambda(s_1(y))$. %
Therefore, the order $\eventless$ on $\universalEvents{\bulk{d}}$ agrees with the order on $\{1,\dotsc,d\}$.
\end{proof}

\begin{defi}\label{def_sculptures}
  A \emph{sculpture}, denoted $\sculpture{Q}{\embedMorphism}{\bulk{d}}$, is an HDA $Q$ together with a bulk $\bulk{d}$ and
  an HDA embedding $\embedMorphism: Q\hookrightarrow \bulk{d}$. 
A \emph{morphism} of sculptures
  $\sculpture{ Q}{\embedMorphism}{ \bulk{ d}}$,
  $\sculpture{ Q'}{\embedMorphism'}{ \bulk{ d'}}$ 
is a pair of HDA
  morphisms $f: Q\to Q'$, $b: \bulk{ d}\to \bulk{ d'}$ such that the
  square
  \begin{equation*}
    \xymatrix{%
      Q \ar[d]_{\embedMorphism} \ar[r]_f & Q' \ar[d]^{\embedMorphism'}
      \\ \bulk{d} \ar[r]^b & \bulk{d'}
    }
  \end{equation*}
  commutes, \ie~$b\circ \embedMorphism= \embedMorphism'\circ f$.

  We say that an HDA $Q$ is \emph{sculptable} if there exists a sculpture $\sculpture{Q}{\embedMorphism}{ \bulk{d}}$.
\end{defi}

For a morphism $( f, b)$ as above, we must have $d'\ge d$ and $b$
injective, hence also $f$ is injective.  Two sculptures are
isomorphic, denoted $\isomorphic$, when $f$ and $b$ are isomorphisms
(implying $d= d'$ and $b= \id$).

For the special case of $Q= Q'$ above, we see that any sculpture
$\sculpture{Q}{\embedMorphism}{\bulk{d}}$  
can be
\emph{over-embedded} into a sculpture
$\sculpture{Q}{b\circ \embedMorphism}{\bulk{d'}}$ 
for $d'> d$.  Conversely, any sculpture
$\sculpture{Q}{\embedMorphism}{\bulk{d}}$ admits a \emph{minimal}
bulk $\bulk{d_\textup{min}}$ for which
$\sculpture{ Q}{\embedMorphism' }{\sculpture{\bulk{d_\textup{min}}}{b'}{ \bulk{d}}}$ with $b'\circ\embedMorphism'=\embedMorphism$,
\ie~such that there is no factorization of the embedding of $Q$ through $\bulk{d'}$ for any
$d'< d_\textup{min}$.  We call such a minimal embedding
\emph{simplistic}.

\begin{rem}
  One precubical set can be seen as sculpted 
  from two different-dim\-ens\-ional bulks, in both cases being a
  simplistic sculpture, \ie~it all depends on the embedding morphism
  (\cf~Figure~\ref{fig:conflict}). Because of this we cannot determine
  from an HDA alone in which sculpture it enters (if any).

  Working with unfoldings is not particularly good either. The
  interleaving square from Figure~\ref{fig:opensquare} (left) can be
  sculpted from $\bulk2$, but its unfolding may be sculpted
  simplistically from
  $\bulk3$ or $\bulk4$; we cannot decide which.

  All the sculptures in Figs.~\ref{fig:conflict}
  and~\ref{fig:opensquare} are simplistic.
\end{rem}

\begin{prop}\label{prop_orderedNeededForSculpting}
	If an HDA $Q$ is not 
ordered, 
then it is not sculptable.
\end{prop}
\begin{proof}
	Any precubical morphism $\embedMorphism:Q\to\bulk{d}$ induces a map
	\[
		\universalEvents{Q}
			\xrightarrow{\universalEvents{\embedMorphism}}
		\universalEvents{\bulk{d}}
			\simeq
		\{1<2<\dots<d\}.
	\]
	If $Q$ is not 
ordered, 
then there exists $a\in \universalEvents{Q}$ such that $a\eventless a$, which implies $\universalEvents{\embedMorphism}(a)\eventless \universalEvents{\embedMorphism}(a)$, which is a contradiction.
\end{proof}

We show that sculptures and regular ordered ST-structures are
in bijective correspondence
while also respecting the computation steps. This result also
resolves the open problem noticed in
\cite[Sec.~3.3]{Johansen16STstruct} that there is no adjunction between
ST-structures and general HDA.

Recall that an ST-structure is regular if it is rooted, connected, and
closed under single events.  Through the
observation from Section~\ref{subsec_Chu_andST} the results in this
section extend to (regular) Chu spaces over $\three$ as well.

\begin{defi}[from ordered regular ST-structures to
  sculptures]\label{def_stintosculptures}
  We define a mapping $\stintosculpture$ that for any regular ordered 
  ST-structure $S$ on events $\evlist E=\{ e_1,\dotsc, e_d\}$,
  generates an HDA, as well as a bulk and an embedding, thus a
  sculpture, $\stintosculpture( S)$, as follows.  By the 
bijection
  between the complete ST-structure $\bulkST{d}$ on events
  $\evlist E$ and $\bulk d$, there is an embedding $S\hookrightarrow 
\bulkST{d} 
\bijection
\bulk d$, where $\hookrightarrow$ simply maps $e_i\in\evlist E$ to $i\in \bulkST{d}$.
  $\stintosculpture( S)$ is given by the composed embedding.
\end{defi}

\begin{defi}[from sculptures to ordered regular ST-structures]
  \label{def_sculptures_to_ST}
  Define a mapping $\sculpintost$ which to a sculpture
  $\sculpture{Q}{\embedMorphism}{\bulk{d}}$ associates the
  ST-structure $\sculpintost(\sculpture{Q}{\embedMorphism}{\bulk{d}})$ as follows.  By the
bijection
between $\bulk d$ and the complete ST-structure $\bulkST{d}$
  on events $\{ 1,\dotsc, d\}$, there is an embedding
  $Q \stackrel{\embedMorphism}{\hookrightarrow} \bulk d
\bijection
\bulkST{d}$.
  $\sculpintost(\sculpture{Q}{\embedMorphism}{\bulk{d}})$ is given by the composed embedding.
\end{defi}

It is clear that
$\sculpintost(\sculpture{Q}{\embedMorphism}{\bulk{d}})$ is rooted, connected and closed under
single events, \ie~regular.

The following result shows a one-to-one correspondence between regular ordered
ST-structures and sculptures; the proof is clear by composition of the
mappings
above.

\begin{thm}
  \label{prop_stSculptst}
  \label{prop_stSculptst-2}
  For any regular ordered ST-structure $\ST$,
  $\sculpintost(\stintosculpture(\ST))\isomorphic \ST$.  For any
  sculpture $\sculpture{Q}{\embedMorphism}{\bulk{d}}$,
  $\stintosculpture(\sculpintost(\sculpture{Q}{\embedMorphism}{\bulk{d}}))\isomorphic
  \sculpture{Q}{\embedMorphism}{\bulk{d}}$.
\end{thm}

We can also understand \sculpintost\ as labeling every cell of the
sculpture with an ST-configuration, or equivalently (because of
Theorem~\ref{prop_STstructChu3}) with a Chu state.

\begin{lem}\label{lem_HsFunctorial}
  The mapping $\stintosculpture$ is functorial, in the sense that an
  ordered ST-morphism $f: \ST_1\to \ST_2$ is translated into an HDA
  morphism
  $\stintosculpture( f): \stintosculpture( \ST_1)\to \stintosculpture(
  \ST_2)$, given by $\stintosculpture( f)({(S,T)})=\stintosculpture({f(S,T)})$.
If $f$
  is total and injective, then $\stintosculpture( f)$ is also a
  sculpture morphism.
\end{lem}

\begin{proof}%
The first part of the lemma is trivial.

For the second part we denote $S_{1}=(E_{1},\sts_{1})$ and $S_{2}=(E_{2},\sts_{2})$. 
The morphism $b:\bulk{d_{1}} \isomorphic \bulkST{\num{E_{1}}} \to \bulkST{\num{E_{2}}} \isomorphic \bulk{d_{2}}$ is defined by the map $f: E_{1} \to E_{2}$, which makes the sculptures morphism diagram commute.
  \begin{equation*}
    \raisebox{\depth}{\xymatrix{%
      \stintosculpture( \ST_1) \ar[d]_{\embedMorphism_{1}}
      \ar[r]_{\stintosculpture( f)} & \stintosculpture( \ST_2)
      \ar[d]^{\embedMorphism_{2}}
      \\ \bulk{d_{1}} \ar[r]^b & \bulk{d_{2}} 
    }} \qedhere
\end{equation*}
\end{proof}

\section{Decidability for the Class of Sculptures}\label{sec_decision}

We proceed to develop an algorithm to decide whether a given HDA can
be sculpted.  
At first one could simply search for embedding into bulks of any dimension limited by the number of edges in the HDA.
But a naive calculation reveals this to be more than doubly exponential in the number of edges.%
\footnote{For an HDA $Q$ with $|Q_1| = n$ it is enough to check for embeddings into the single bulk of the largest dimension $n$, because any sculpture can be over-embedded. There are $|\bulk n|^{|Q|}$ maps to check, which is larger than focusing on maps between transitions only, i.e., larger than $|\bulk n _{1}|^{n}=(n*2^{(n-1)})^{n}$. This should also be multiplied with the amount of time it takes to check whether an individual map is an embedding, i.e., checking injectivity, cubical laws, face maps preservation for all higher cells, etc.
}
In this section we work out a more algorithmic approach which is also more efficient.
First we define a way of translating HDA into ST-structures without the need of a bulk and an embedding.
Instead we give an inductive construction that works with rooted paths.

\begin{defi}
	A path having the following form, for $v_{i}\in Q_{0}$ and $e_{j}\in Q_{1}$, 
	\begin{equation}\label{eq:SequentialPath}
		v_0 \xrightarrow{s} e_1 \xrightarrow{t} v_1 \xrightarrow{s} e_2 \xrightarrow{t} v_2 \xrightarrow{s} \dots \xrightarrow{s} e_n \xrightarrow{t} v_n	
	\end{equation}
	will be called sequential.	
	An HDA $Q$ has \emph{non-repeating events} if for every sequential path
	all universal labels $\lambda(e_1), \lambda(e_2),\dotsc,\lambda(e_n)$ are different.
\end{defi}

\begin{prop}
	If $Q$ has repeating events, then it cannot be sculpted.
\end{prop}
\begin{proof}
	If a path $\pi$ in $Q$ repeats events, then its image $em(\pi)$ in $\bulk{d}$ also repeats events, by the functoriality of $\universalEvents{}$;
but bulks have non-repeating events.
\end{proof}

\begin{prop}
	If $Q$ has non-repeating events, then it 
is consistent.
\end{prop}
\begin{proof}
	Easy.
\end{proof}

\noindent\begin{minipage}[]{.75\textwidth}
\begin{exa}
An HDA that has repeating events is the full square, $q$, to the right which has the upper-left and lower-right corners identified into $q_{0}$. 
We find the sequential path
$I \xrightarrow{s} s_{2}(q) \xrightarrow{t} q_{0} \xrightarrow{s} t_{2}(q) \xrightarrow{t} q'_{0}$ on which the same label $\lambda_{1}(q)$ appears twice.
This example is acyclic; otherwise, the non-repeating property implies acyclicity.
\end{exa}
\end{minipage}%
\begin{minipage}[]{.25\textwidth}
  \centering
  \begin{tikzpicture}[>=stealth']
    \path[fill=black!10] (0,0) to (2,0) to (2,2) to (0,2) to (0,0);
    \node[state] (00) at (0,0) {};
    \node[state] (10) at (2,0) {};
    \node[state] (01) at (0,2) {};
    \node[state] (11) at (2,2) {};
    \path (00) edge (01);
    \path (00) edge (10);
    \path (01) edge (11);
    \path (10) edge (11);
    \node at (1,1) {$q$};
    \node at (-.2,-.35) {$I$};
    \node at (-.2,2.35) {$q_{0}$};
    \node at (2.2,-.35) {$q_{0}$};
    \node at (2.2,2.35) {$q'_{0}$};
    \node at (1,.25) {$\vphantom{t}s_2 q$};
    \node at (1,1.75) {$t_2 q$};
  \end{tikzpicture}
\end{minipage}

\begin{defi}[from HDA to ST-structures through paths]\label{def_hdaTOst}
	Define a map $\hintost:\allHDA\rightarrow\allST$ which builds an
  ST-structure $\hintost(Q)=(\universalEvents{Q}, \sts)$ in the following way. For every path $\pi=\pi'\xrightarrow{\alpha} q$ we assign an ST-configuration $\hintost(\pi)=(S_\pi, T_\pi)$ in the following way.
  \begin{enumerate}
  \item For the minimal rooted path we associate
    $\hintost(\initialCell)=(\emptyset,\emptyset)$.
  \item If $\alpha=s_i$, then we put
    $\hintost(\pi)=\hintost(\pi')\cup (\{\lambda_i(q)\},
    \emptyset)=(S_\pi\cup\{\lambda_i(q)\},T_\pi)$, \ie~we start the
    event $\lambda_i(q)$.
  \item If $\alpha=t_i$, then we put
    $\hintost(\pi)=\hintost(\pi')\cup (\emptyset,
    \{\lambda_i(\finishPath{\pi'})\})$, \ie~we terminate the event
    $\lambda_i(\finishPath{\pi'})$.
  \end{enumerate}
  Finally, $\sts$ is the set of all these ST-configurations, \ie 
  \[
  		\hintost(Q)=\bigcup_{\pi\in \Path(Q)_*} \hintost(\pi)
  \]
where $\Path(Q)_*$ denotes the set of all rooted paths of $Q$.
\end{defi}

The construction is similar to an \emph{unfolding}~\cite{DBLP:conf/calco/FahrenbergL15}; 
see~\cite[Def.~3.39]{Johansen16STstruct} for a related construction.

The next lemmas are used to establish that for every path $\pi$ the pair $(S_\pi, T_\pi)$ is indeed an ST-configuration.

\begin{lem}\label{lem:STForSequntialPaths}
	If $\pi$ is a sequential path {\normalfont(\ref{eq:SequentialPath})}, then
	\[
		S_\pi=T_\pi=\{\lambda(e_1),\dotsc,\lambda(e_n)\}.
	\]
\end{lem}
\begin{proof}
	Obvious induction.
\end{proof}

\begin{lem}\label{lem:HomotopicPaths}
	Homotopic paths have the same associated ST-configurations (\ie~if $\pi\sim\varrho$, then $S_\pi=S_\varrho$, $T_\pi=T_\varrho$).
\end{lem}
\begin{proof}
	It is enough to consider the case when $\pi$ and $\varrho$ are elementary homotopic and that the homotopy changes the final segments of these paths. Thus
    \[
        \pi = \sigma \xrightarrow{\alpha_i} q \xrightarrow{\beta_j} r,\qquad \varrho = \sigma \xrightarrow{\beta_k} q' \xrightarrow{\alpha_l} r,
    \]
    where one of the following cases holds (denote $s=\finishPath{\sigma}$).
    \begin{itemize}
    \item %
        $\alpha=\beta=s$, $i<j$, $k=j-1$, $l=i$. Then $T_\pi=T_\sigma=T_\varrho$ and
        \begin{align*}
            S_\pi&= S_\sigma \cup \{\lambda_i(q)\} \cup \{\lambda_j(r)\} = S_\sigma \cup \{\lambda_i(s_j(r))\} \cup \{\lambda_j(r)\} = S_\sigma \cup \{\lambda_i(r),\lambda_j(r)\}\\
            S_\varrho&= S_\sigma \cup \{\lambda_{j-1}(q')\} \cup \{\lambda_i(r)\} = S_\sigma \cup \{\lambda_{j-1}(s_i(r))\} \cup \{\lambda_i(r)\} = S_\sigma \cup \{\lambda_i(r),\lambda_j(r)\}.
        \end{align*}
    \item %
        $\alpha=s$, $\beta=t$, $i>k=j$, $l=i-1$. Then 
        \begin{align*}
            S_\pi&=S_\sigma\cup \{\lambda_i(q)\}=S_\sigma\cup \{\lambda_{i-1}(t_j(q))\} = S_\sigma\cup \{\lambda_l(r)\} = S_\varrho,\\
            T_\pi&=T_\sigma\cup \{\lambda_j(q)\}=T_\sigma\cup \{\lambda_{j}(s_i(q))\} = T_\sigma\cup\{\lambda_k(s)\} = T_\varrho.
        \end{align*}
    \item %
        $\alpha=s$, $\beta=t$, $i=l<j$, $k=j-1$.
    \item %
        $\alpha=\beta=t$, $k=j<i$, $l=i-1$.
    \end{itemize}
    Calculations in the last two cases are similar. 
\end{proof}

\begin{lemC}[{\cite[Lem.4.38]{Uli05PhD}}]\label{lem:HomotopicPathsType}
	Every rooted path is homotopic to a path of the type $(st)^k
        s^n$, for 
any $n,k\geq 0$.
\end{lemC}
\begin{proof}
	Assume that $\pi$ has a segment of type $sst$, namely
	\[
		s_i(s_j(r))\xrightarrow{s_i} s_j(r) \xrightarrow{s_j} r \xrightarrow{t_l} t_l(r).
	\]
	If $j=l$, then we replace it by a homotopic segment
	\begin{itemize}
	\item
		$s_i(s_j(r))=s_{j-1}(s_i(r))\xrightarrow{s_{j-1}} s_i(r) \xrightarrow{s_i} r \xrightarrow{t_l} t_l(r)$ if $i<j$
	\item
		$s_i(s_j(r))=s_{j}s_{i+1}(r)\xrightarrow{s_{j}} s_{i+1}(r) \xrightarrow{s_{i+1}} r \xrightarrow{t_l} t_l(r)$ if $i\geq j$	
	\end{itemize}
	to assure that $j\neq l$. Next, we replace it by
	\begin{enumerate}
	\item
		$s_i(s_j(r))\xrightarrow{s_i} s_j(r) \xrightarrow{t_{l-1}} t_{l-1}(s_j(r))=s_j(t_l(r)) \xrightarrow{s_j} t_l(r)$ if $j<l$,
	\item
		$s_i(s_j(r))\xrightarrow{s_i} s_j(r) \xrightarrow{t_{l}} t_{l}(s_j(r))=s_{j-1}(t_l(r)) \xrightarrow{s_{j-1}} t_l(r)$ if $j>l$,
	\end{enumerate}
	and obtain a homotopic segment of type $sts$. We repeat this procedure as long as there is a type $sst$ subpath (which is finitely many times). Eventually  we obtain a path homotopic to $\pi$ having the required type.
\end{proof}

\begin{lem}
	Fix $n\ge 1$ and $i\in \{1,\dotsc,n\}$. Every path of type $s^n$, 
starting in a vertex, %
is homotopic to a path of type $s_1 s_2 \dots s_{n-2}s_{n-1} s_i$.
\end{lem}
\begin{proof}
	For $i=n$ this follows from the canonical presentation of an iterated face map. In general, we start with $s_1\dots s_n$ and move $s_i$ to the rightmost place using precubical identities,
i.e., $s_1 \dots  s_{i-1} s_i  s_{i+1} s_{i+2} \dots s_n =
s_1 \dots  s_{i-1} s_i  s_i s_{i+2}  \dots s_n =
s_1 \dots  s_{i-1} s_i  s_{i+1} s_i s_{i+3}  \dots s_n = \dots = s_1 \dots  s_{n-2} s_i s_n =
s_1 \dots  s_{n-2} s_{n-1} s_i$.
\end{proof}

\begin{lem}\label{lem:CanonicalPath}
  For every $n,k\ge 0$, every rooted path $\pi$ ending in $q\in Q_n$,
  and any $i\in\{1,\dotsc,n\}$, there exists a path that is homotopic
  to $\pi$ and has the type
  $(s_1 t_1)^k s_1 s_2 \dots s_{n-2}s_{n-1} s_i$.
\end{lem}
\begin{proof}
	This follows from the two preceding lemmas.
\end{proof}

\begin{prop}
  For every rooted %
path $\pi$, $(S_\pi, T_\pi)$ is an ST-configuration (\ie~$S_\pi\supseteq T_\pi$).
\end{prop}
\begin{proof}
	Induction with respect to $\pi$; enough to check the case when $\pi=\pi'\xrightarrow{t_i} t_i(q)$, $q=\finishPath{\pi'}$. By Lemma \ref{lem:CanonicalPath} there exists a path $\varrho'=\varrho'' \xrightarrow{s_i} q$ homotopic to $\pi'$. Thus,
	\[
		T_\pi = T_{\pi'}\cup \{\lambda_i(q)\} \overset{L.\ref{lem:HomotopicPaths}}= T_{\varrho'}\cup \{\lambda_i(q)\}=T_{\varrho''}\cup \{\lambda_i(q)\}
		\overset{ind}\subseteq S_{\varrho''}\cup \{\lambda_i(q)\}=S_{\varrho'}=
S_\pi. \qedhere
	\]
\end{proof}

\begin{prop}\label{prp:ActiveEvents}
    Assume that an HDA $Q$ has non-repeating events. Then $S_\pi\setminus T_\pi = \lambda(\finishPath{\pi})$ for every rooted path $\pi$. 
\end{prop}
\begin{proof}
Note that $\lambda(\finishPath{\pi})$ is a tuple, thus a set of universal labels with an order on them; and similarly, the set of events on the left are ordered.
    We use induction with respect to the structure of $\pi$. If $\pi=\initialCell$ --- obvious. By Lemmas \ref{lem:CanonicalPath} and \ref{lem:HomotopicPaths}  we can assume that $\pi$ has the type $(s_1t_1)^l s_1\dots s_n$. Denote $q=\finishPath{\pi}\in Q_n$. Consider two cases:
    \begin{itemize}
    \item
    	$n=0$. Then $\pi=\pi'\xrightarrow{t_1}q$, with $\lambda(q)=\emptyset$. Using the inductive hypothesis we obtain
    	\[
    		S_\pi \stackrel{\mathit{def}}{=} S_{\pi'} \stackrel{\mathit{ind}}{=} T_{\pi'}\cup \{\lambda(\finishPath{\pi')}\} \stackrel{\mathit{def}}{=} T_\pi.
    	\]
    \item
    	$n>0$. Let $\pi'$ be the prefix of $\pi$ of length $2l$. Since $\finishPath{\pi'}$ is a state, then $S_{\pi'}=T_{\pi'}$ (by the previous case) and $S_{\pi}=S_{\pi'}\cup \lambda(q)$, $T_\pi=T_{\pi'}$ (by Definition~\ref{def_hdaTOst}). It remains to show that $\lambda(q)\cap S_{\pi'}=\emptyset$. For every $i$ 
    	\[
    		\lambda(s_1\dots s_{i-1} s_{i+1}\dots s_n(q))=\lambda_i(q)
    	\]
    	and the path $\pi'\xrightarrow{s_1} s_1\dots s_{i-1} s_{i+1}\dots s_n(q)\xrightarrow{t_1} s_1\dots s_{i-1} t_i s_{i+1}\dots s_n(q)$ 
    	is sequential. Since $Q$ has non-repeating events, Lemma \ref{lem:STForSequntialPaths} implies that $\lambda_i(q)\not\in S_{\pi'}$. \qedhere
    \end{itemize}
\end{proof}

\begin{cor}\label{cor_newEvents}
	Assume that $Q$ has non-repeating events. Then for every $i$ and $q\in Q$ and every rooted path $\pi\in\Path(Q)_*$ that can be extended to $\pi\xrightarrow{s_i}q$, then $\lambda_i(q)\not\in S_\pi$.
\end{cor}
\begin{proof}
	Otherwise, $S_{\pi\xrightarrow{s_i}q}=S_\pi$, which implies that $\lambda(s_i(q))=\lambda(q)$.
\end{proof}

\begin{prop}
	If $Q$ has non-repeating events, then $\hintost(Q)$ is a regular ST-structure.
\end{prop}
\begin{proof}
		Conditions (A) and (B) are obvious. To prove (C), fix $\pi\in\Path(Q)_*$ and
		\[
			e=\lambda_i(q)\in \lambda(q)=S_\pi\setminus T_\pi,
		\]
		where $q=\finishPath{\pi}$.
		Let $\varrho$ be a rooted path such that $\varrho\xrightarrow{s_i}q$ is homotopic to $\pi$
		and let $\varrho'=\pi\xrightarrow{t_i} t_i(q)$.
		Then
		\[
			(S_\pi, T_\pi \cup \{e\})
			=
			(S_\pi, T_\pi \cup \lambda_i(q))
			=
			(S_{\varrho'},T_{\varrho'})
			\in
			\hintost(Q)
		\]
		and
		\[
			(S_\pi,T_\pi)=(S_\varrho\cup \{\lambda_i(q)\},T_\varrho)=(S_\varrho\cup \{e\},T_\varrho).
		\]
		But $e\not\in S_\varrho$ by Corollary \ref{cor_newEvents}, so $(S_\varrho,T_\varrho)=(S_\pi\setminus\{e\},T_\pi)\in\hintost(Q)$.
\end{proof}

If the HDA in question is a sculpture, then there is a natural
equivalence relation on its cells which captures the notion of
\emph{event} better than the universal event labelling.

\begin{defi}
  \label{def_eventEquivSculpt}
  \label{cor_eventEquivSculpt}
  For a sculpture $\sculpture{Q}{\embedMorphism}{\bulk{d}}$, define
  $\universalEvents{\embedMorphism}:\universalEvents{Q}\to
  \universalEvents{\bulk{d}}\simeq \universalEvents{\bulkST{d}}\simeq
  \{1,\dotsc,d\}$ to be the map induced by $\embedMorphism$, \ie~$\universalEvents{\embedMorphism}(\lambda(q))=\lambda(\embedMorphism(q))$.
  This induces an equivalence relation on \universalEvents{Q} which we
  denote $\eventEquivSculpt$, \ie~$\lambda(q)\eventEquivSculpt\lambda(q')$ iff
  $\universalEvents{\embedMorphism}(\lambda(q))=\universalEvents{\embedMorphism}(\lambda(q'))$.
\end{defi}

\begin{lem}\label{lem:SpiSculpt}
	Let $\sculpture{Q}{\embedMorphism}{\bulkST{d}}$ be a sculpture, then $\quotientofwrt{(S_\pi, T_\pi)}{\eventEquivSculpt}=\embedMorphism(\finishPath{\pi})$ for every rooted path $\pi$.
\end{lem}
\begin{proof}
Note that we work with a sculpture that is embedded directly in the $\bulkST{d}$, which is isomorphic to $\bulk{d}$. We do this to simplify the arguments, for otherwise we would have had to go through the isomorphism using the \chuintost\ to translate between tuples and ST-configurations.

We use induction on the length of the path $\pi$. By Lemmas \ref{lem:CanonicalPath} and \ref{lem:HomotopicPaths} we may assume that $\pi$ has type $(s_1t_1)^ns_1\dots s_k$, $q=\finishPath{\pi}\in Q_k$. 
For $\pi=\initialCell$, $(S_\pi, T_\pi) = (\emptyset, \emptyset)= \embedMorphism(\initialCell)$. 

For $k=0$, then $\pi=\pi'\transition{s_{1}}e\transition{t_{1}}q$ is sequential. By Definition~\ref{def_hdaTOst} $S_\pi = S_{\pi'} \cup\{\lambda(e)\}$ and $T_\pi = T_{\pi'}\cup\{\lambda(e)\}$ with $\lambda(e)\not\in S_{\pi'}$ (because of Corollary~\ref{cor_newEvents} since a sculpture has non-repeating events).
On the right side, the path $\embedMorphism(\pi)=\embedMorphism(\pi')\transition{s_{1}}\embedMorphism(e)\transition{t_{1}}\embedMorphism(\finishPath{\pi})$ has all universal labels different too, thus $\lambda(\embedMorphism(e))\not\in S'$ for $(S',T')=\embedMorphism(\finishPath{\pi'})$, and by construction $(S'\cup\{\lambda(\embedMorphism(e))\},T'\cup\{\lambda(\embedMorphism(e))\})=\embedMorphism(\finishPath{\pi})$.
We finish this case by applying the induction hypothesis to obtain
$\quotientofwrt{(S_{\pi}, T_{\pi})}{\eventEquivSculpt}=(\quotientofwrt{S_{\pi'}}{\eventEquivSculpt}\cup\{\universalEvents{\embedMorphism}(\lambda(e))\},\quotientofwrt{T_{\pi'}}{\eventEquivSculpt}\cup\{\universalEvents{\embedMorphism}(\lambda(e))\})\stackrel{ind}{=}
(S'\cup\{\universalEvents{\embedMorphism}(\lambda(e))\},T'\cup\{\universalEvents{\embedMorphism}(\lambda(e))\})\stackrel{\ref{def_eventEquivSculpt}}{=}(S'\cup\{\lambda(\embedMorphism(e))\},T'\cup\{\lambda(\embedMorphism(e))\})$.

For $k>0$, then
$\pi=\pi'\transition{s_{1}}q_{1}\dots\transition{s_{k}}q=\pi''\transition{s_{k}}q$
with $\pi'$ sequential.  From Proposition~\ref{prp:ActiveEvents} we
have $S_{\pi}=T_{\pi}\cup\{\lambda(q)\}$ and
$T_{\pi''}=T_{\pi}=T_{\pi'}=S_{\pi'}$. Denote by
$(S,T)=\embedMorphism(\finishPath{\pi})=\embedMorphism(q)$ and by
$(S'',T'')=\embedMorphism(\finishPath{\pi''})=\embedMorphism(s_{k}(q))$. By
construction, $(S,T)=(S''\cup\{\lambda_{k}(q)\},T'')$ and by the
induction hypothesis we have
$(S''\cup\{\lambda_{k}(q)\},T'') \stackrel{\text{ind}}{=}
(\quotientofwrt{S_{\pi''}}{\eventEquivSculpt}\cup\{\lambda_{k}(q)\},\quotientofwrt{T_{\pi''}}{\eventEquivSculpt})
=
(\quotientofwrt{S_{\pi'}}{\eventEquivSculpt}\cup\{\lambda(s_{k}(q))\}\cup\{\lambda_{k}(q)\},\quotientofwrt{T_{\pi'}}{\eventEquivSculpt})
=
(\quotientofwrt{S_{\pi'}}{\eventEquivSculpt}\cup\{\lambda(q)\},\quotientofwrt{T_{\pi'}}{\eventEquivSculpt})
= (\quotientofwrt{S_{\pi}}{\eventEquivSculpt},\quotientofwrt{T_{\pi}}{\eventEquivSculpt})$,
since by the consistency and non-repeated events properties we know
that
$\lambda_{k}(q)\not\in\quotientofwrt{S_{\pi''}}{\eventEquivSculpt}$.
\end{proof}

\begin{prop}
  \label{pr:sts=stpi}
 For a (connected) simplistic sculpture $\sculpture{Q}{\embedMorphism}{\bulk{n}}$ we have 
\[
  \sculpintost(\sculpture{Q}{\embedMorphism}{\bulk{n}}) \isomorphic
  \quotientofwrt{\hintost(Q)}{\eventEquivSculpt}\,.
\]
\end{prop}

\begin{proof}%
Note that the requirement of being simplistic is only needed in order to have the same set of events on both sides.
The events generated on the left side by \sculpintost\ are $\{1,\dots,n\}=\universalEvents{\bulk{n}}$ which are the events obtained on the right side due to the application of \eventEquivSculpt, having the same order.

The isomorphism is then exhibited by the identity map $f$ on the above sets of events.
Showing that $f$ preserves ST-configurations is easy by using the previous Lemma~\ref{lem:SpiSculpt} since every ST-configuration is generated as $\embedMorphism(q)$, but since all cells are reachable then there exists a path $\pi$ ending in $q$ so that we need to show $f(\embedMorphism(\finishPath{\pi}))=\quotientofwrt{(S_{\pi},T_{\pi})}{\eventEquivSculpt}$, which is done by the lemma.
\end{proof}

\begin{cor}
For a sculpture $\sculpture{Q}{\embedMorphism}{\bulkST{d}}$ the equivalence $\eventEquivSculpt$ is non-collapsing.
\end{cor}

\begin{defi}\label{def_properEventIdent}
	Let $Q$ be an HDA. \emph{A proper event identification} on $Q$ is an equivalence relation $\evEqRel$ on $\universalEvents{Q}$ such that
	\begin{enumerate}
	\item\label{def_properEventIdentAnti}
		The quotient preorder on $\universalEvents{Q}/_\evEqRel$ induced from $\universalEvents{Q}$ is antisymmetric. Equivalently, if $a\eventless b$, $c\eventless d$, $a\evEqRel d$, $b\evEqRel c$ for $a,b,c,d\in\universalEvents{Q}$, then $a\evEqRel b \evEqRel c\evEqRel d$. 
	\item\label{def_properEventIdentFunction}
		If $\finishPath{\pi}=\finishPath{\pi'}$, then $(S_\pi/_\evEqRel,T_\pi/_\evEqRel)=(S_{\pi'}/_\evEqRel,T_{\pi'}/_\evEqRel)$.
	\item\label{def_properEventIdentInjective}
		If $(S_\pi/_\evEqRel,T_\pi/_\evEqRel)=(S_{\pi'}/_\evEqRel,T_{\pi'}/_\evEqRel)$, then $\finishPath{\pi}=\finishPath{\pi'}$.
	\end{enumerate}
\end{defi}

\begin{rem}\label{remark_ProperEvents}
	An equivalence relation $\evEqRel$ is a proper event identification if the sequence of relations
	\[
		Q \xrightarrow{en^{-1}} \Path(Q)_* \xrightarrow{\hintost} \hintost(Q)\xrightarrow{\subseteq} \bulkST{\num{\universalEvents{Q}}} \to \bulkST{\num{\universalEvents{Q}/_\evEqRel}}
	\]
	forms an injective function. Note that the right-most map is not (in general) an ST-map.
\end{rem}

\begin{lem}\label{lem:TrivialIdent}
	Let $\evEqRel$ be a proper event identification on $Q$. Then for every $\pi\in\Path(Q)_*$,  $\evEqRel$ is trivial when restricted to $S_\pi$.
\end{lem}
\begin{proof}
	Assume that there exists a path $\pi\in \Path(Q)_*$ and $a\neq b\in S_\pi$ such that $a\evEqRel b$. Without loss of generality we may assume that $\finishPath{\pi}$ is a state (extending $\pi$ with some $t_1$--type segments if needed), and also that $\pi$ is sequential (by Lemma \ref{lem:HomotopicPaths} and \ref{lem:CanonicalPath}). Denote
	\[
		\pi = v_0 \xrightarrow{s} e_1 \xrightarrow{t} v_1 \xrightarrow{s} e_2 \xrightarrow{t} v_2 \xrightarrow{s} \dots \xrightarrow{s} e_n \xrightarrow{t} v_n	
	\]
	and let $\pi_j$ denote the prefix of $\pi$ ending at $v_j$. Since $S_\pi=\{\lambda(e_i)\}_{i=1}^n$, then there exist $k<l$ integers such that $\lambda(e_k)\evEqRel \lambda(e_l)$. But then
	\[
		(S_{\pi_l}/_\evEqRel,T_{\pi_l}/_\evEqRel)=(S_{\pi_{l-1}}/_\evEqRel, T_{\pi_{l-1}}/_\evEqRel),
	\]
	so $\evEqRel$ is not a proper identification, breaking \ref{def_properEventIdent}(\ref{def_properEventIdentInjective}).
\end{proof}

\begin{cor}
A proper event identification equivalence on $Q$ is non-collapsing.
\end{cor}

\begin{thm}\label{th:sculptalgo}
  Let $Q$ be a connected HDA with non-repeating events. The following are
  equivalent:
	\begin{enumerate}
	\item
		$Q$ can be sculpted.
	\item
		There exists a proper event identification on $Q$.
	\end{enumerate}
\end{thm}
\begin{proof}
	\def\numbMap{j}
	(2)$\Rightarrow$(1).
	Let $\evEqRel$ be a proper event identification equivalence. Fix an order-preserving bijective map $\numbMap:\universalEvents{Q}/_\evEqRel \to \{1,\dotsc,d\}$, which exists by Definition~\ref{def_properEventIdent}(\ref{def_properEventIdentAnti}), with $d$ being the dimension of the quotient.
For $q\in Q_n$ choose a path $\pi$ ending at $q$ and put
	\[
		\embedMorphism(q)=(\numbMap(S_{\pi}/_\evEqRel),\numbMap(T_{\pi}/_\evEqRel)).
	\]
Lemma~\ref{lem:TrivialIdent} tells that the equivalence classes in $S_{\pi}/_\evEqRel $ and $T_{\pi}/_\evEqRel$ are singletons, thus making $\embedMorphism(q)=(\numbMap(S_{\pi}),\numbMap(T_{\pi}))$. We abused the notation here, as the map $\numbMap$ should be applied to an equivalence class, but instead we apply it to elements of \universalEvents{Q}, since for a particular path as we have here, the equivalence classes are singletons.
	Condition \ref{def_properEventIdent}(\ref{def_properEventIdentFunction}) assures that $\embedMorphism(q)$ does not depend on the choice of $\pi$ as long as $\finishPath{\pi}=q$. It remains to prove that $\embedMorphism$ is a morphism of precubical sets, \ie~it preserves the precubical maps. 
	
	Let $q\in Q_n$, $i\in \{1,\dotsc,n\}$, 
$\pi=\pi'\xrightarrow{s_i}q$;
then $\embedMorphism(q)=(j(S_\pi),j(T_\pi))$.
Since $\numbMap$ is injective and order-preserving, we have
	\begin{multline*}
		s_i(\embedMorphism(q))=
		s_i(\numbMap(S_\pi),\numbMap(T_\pi))\overset{\text{P.\ref{prp:ActiveEvents}}}=
		s_i(\numbMap(T_\pi)\cup \numbMap(\lambda(q)),\numbMap(T_\pi))= \\
		(\numbMap(T_\pi)\cup \numbMap(\lambda(q)\setminus\{\lambda_i(q)\}),\numbMap(T_\pi))=
		(\numbMap(T_\pi)\cup \numbMap(\lambda(s_i(q))),\numbMap(T_\pi))= \\
		(\numbMap(T_{\pi'})\cup \numbMap(\lambda(s_i(q))),\numbMap(T_{\pi'}))=
		(\numbMap(S_{\pi'}),\numbMap(T_{\pi'}))=\embedMorphism(s_i(q)).
	\end{multline*}
	Now let $\pi''=\pi\xrightarrow{t_i} t_i(q)$. A similar calculation shows that
	\[
		t_i(\embedMorphism(q))=t_i(j(S_\pi), j(T_\pi))= (j(S_{\pi''}), j(T_{\pi''}))=\embedMorphism(t_i(q)).
	\]
	As a consequence, $\embedMorphism$ is a precubical map $Q\to \bulkST{d}$.
	
	(1)$\Rightarrow$(2). Let $\sculpture{Q}{\embedMorphism}{\bulkST{d}}$ be a sculpture involving $Q$. Consider the equivalence relation \eventEquivSculpt\  from Definition~\ref{def_eventEquivSculpt}. Since this was defined using the functor \universalEvents{} then we know that it preserves the order, thus respecting property \ref{def_properEventIdent}(\ref{def_properEventIdentAnti}) of being a proper event identification.
The other two properties are derived using Lemma~\ref{lem:SpiSculpt}.
For two paths with $\finishPath{\pi}=q=\finishPath{\pi'}$ we have $\quotientofwrt{(S_{\pi},T_{\pi})}{\eventEquivSculpt}\stackrel{L.\ref{lem:SpiSculpt}}{=}\embedMorphism(\finishPath{\pi})=\embedMorphism(\finishPath{\pi'})\stackrel{L.\ref{lem:SpiSculpt}}{=}\quotientofwrt{(S_{\pi'},T_{\pi'})}{\eventEquivSculpt}$.
For the last property, start with $\embedMorphism(\finishPath{\pi})\stackrel{L.\ref{lem:SpiSculpt}}{=}\quotientofwrt{(S_{\pi},T_{\pi})}{\eventEquivSculpt}=\quotientofwrt{(S_{\pi'},T_{\pi'})}{\eventEquivSculpt}\stackrel{L.\ref{lem:SpiSculpt}}{=}\embedMorphism(\finishPath{\pi'})$, and because of the injectivity of \embedMorphism\ we have $\finishPath{\pi}=\finishPath{\pi'}$.
\end{proof}

Since the number of equivalence relations on $\universalEvents{Q}\times \universalEvents{Q}$ is finite,
then Theorem~\ref{th:sculptalgo} translates into an \emph{algorithm} to
determine whether $Q$ is a sculpture: First apply $\hintost(Q)$; then choose some equivalence relation on $\universalEvents{Q}\times \universalEvents{Q}$ and check whether it is a proper event identification.
The dimension $m=\num{\universalEvents{Q}}$ is smaller than or equal (when there is no concurrency) to the number of edges $\num{Q_{1}}=n$. Therefore, the number of relations on $\universalEvents{Q}$ that need to be checked is $2^{m^{2}} < 2^{n^{2}}$, which in the worst case can be more than exponential in the number of edges of $Q$. For each relation we need to check both that it is an equivalence and the proper event identification properties. If we know how to pick only the equivalence relations, which are exponential in number (i.e., using the Bell numbers\footnote{See the Bell numbers sequence as \url{https://oeis.org/A000110} in the OEIS.} they are exactly $(\frac{m}{e\cdot ln\,m})^{m} < B_{m} < (\frac{0.792 \cdot m}{ln\,(m+1)})^{m} $, see \cite{berend2010improved}) then we have to check these only for proper event identification. But we can do better by constructing a proper event identification (when it exists) while we traverse the HDA with the $\hintost$.
\label{pageFirstAlgComplexity}

In the following we give a more intuitive algorithm, using
constructions which iteratively \emph{repair} $\hintost(Q)$ by
constructing a finite sequence of increasing equivalence relations, in the end reaching a proper event identification.

For an HDA $Q$, using the notation of Def.~\ref{def_hdaTOst}, let
$\rho_0\subseteq Q\times \hintost(Q)$ be the relation
$\rho_0=\{( q, \hintost( \pi))\mid \finishPath{\pi}= q\}$. 
We call this an \emph{ST-labeling}, forming the composition of the first two relations from Remark~\ref{remark_ProperEvents}. 
For an equivalence relation $\mathord\sim\subseteq \universalEvents{Q}\times \universalEvents{Q}$, let
$\rho_\sim=\{( q, \quotientofwrt{( S, T)}{\sim})\mid( q,( S, T))\in
  \rho_0\}$.

First, the following lemma
shows that we can restrict our attention to only ST-labellings of $0$-cells
(and because of the previous results, it is enough to apply \hintost\ only to sequential paths).%

\begin{lem}
  \label{le:hda to sculp zero is enough}
  If $|\{ \sigma\mid( q, \sigma)\in \rho_0\}|> 1$ for some $q\in Q_k$,
  $k\ge 1$, then also $|\{ \sigma'\mid( q', \sigma')\in \rho_0\}|> 1$
  for some $q'\in Q_0$.
\end{lem}

\begin{proof}%
  Assume $|\{ \sigma\mid( q, \sigma)\in \rho_0\}|> 1$, then there exist two different paths $\pi,\pi'$ ending in $q$ s.t.\
  $( q, ( S_{\pi}, T_{\pi}))$, $( q,( S_{\pi'}, T_{\pi'}))\in \rho_0$ with 
  $T_{\pi}\ne T_{\pi'}$; this being the only case since, by Proposition~\ref{prp:ActiveEvents}, $S_{\pi}\setminus T_{\pi}=\lambda(\finishPath{\pi})=S_{\pi'}\setminus T_{\pi'}$.  
We can complete both $\pi$ and $\pi'$ by the same sequence of t-steps from $q$ to its upper corner, \ie~there exist
$\pi_{0}=\pi\transition{t_{k}}\dots\transition{t_{1}}q_{0}$ and 
$\pi'_{0}=\pi'\transition{t_{k}}\dots\transition{t_{1}}q_{0}$,
with $\finishPath{\pi_{0}},\finishPath{\pi'_{0}}\in Q_{0}$.
By definition $T_{\pi_{0}}=T_{\pi}\cup\lambda(q)$ and $T_{\pi'_{0}}=T_{\pi'}\cup\lambda(q)$, which are different, meaning that we found the state $q_{0}$ for the lemma.
\end{proof}

We can immediately rule out ST-labellings in which a cell receives
ST-configurations with different numbers of events (this would break the property \ref{def_properEventIdent}(\ref{def_properEventIdentFunction}) of a proper event identification in an irreparable way, cf.~Lemmas~\ref{lem:TrivialIdent} and \ref{lem:STForSequntialPaths}):

\begin{lem}
  \label{le:algo no diff numbers}
  If there is $q\in Q_0$ and $( q,( S, S)),( q,( S', S'))\in \rho_0$
  with $| S|\ne| S'|$, then $Q$ cannot be sculpted.
\end{lem}

\begin{proof}%
  Assume to the contrary that $Q$ is in a sculpture with embedding
  $\embedMorphism: Q\to \bulk{n}$.  By construction of $\rho_0$, there
  are two rooted paths $\pi$, $\pi'$ in $Q$ which both end in $q$, but
  with different lengths.  By injectivity of $\embedMorphism$, the
  images of these paths under $\embedMorphism$, here denoted
  $\embedMorphism( \pi)$ and $\embedMorphism( \pi')$, are paths in
  $\bulk{n}$ from the initial state to $\embedMorphism( q)$.  But
  inside the bulk $\embedMorphism( \pi)$ and $\embedMorphism( \pi')$
  are homotopic, in contradiction to them having different lengths.
\end{proof}

We will inductively construct equivalence relations $\sim_n$, with the
property that $\mathord\sim_n\subsetneq \mathord\sim_{ n+ 1}$.  This
procedure will either lead to a relation
$\mathord\sim_N = \,\evEqRel$ that is a proper event identification equivalence as required in Theorem~\ref{th:sculptalgo}
or to an irreparable conflict as explained below.

Let $\mathord\sim_1= \{(\lambda(q),\lambda(q)) \mid \lambda(q)\in\universalEvents{Q}\}$,
the minimal equivalence relation on $\universalEvents{Q}$. 
If $Q$ is a sculpture, then
$\mathord\sim_1\subseteq \mathord{\eventEquivSculpt}$,
hence we can safely start our procedure with $\sim_1$.
Moreover, this minimal equivalence is antisymmetric (preserving the order of $Q$).

Assume, inductively, that $\sim_n$ has been constructed for some
$n\ge 1$.  The next lemma shows that if there are two different cells
which receive the same labeling under $\rho_{ \sim_n}$, then either $Q$
is not a sculpture or we need to backtrack%
, \ie~we would break property \ref{def_properEventIdent}(\ref{def_properEventIdentInjective}) since we have equated too much.

\begin{lem}
  \label{le:hda to sculpt rho not inj}
  If there are $( q, \sigma),( q', \sigma)\in \rho_{ \sim_n}$ with
  $q\ne q'$, and $Q$ can be sculpted, then $\mathord\sim_n\not\subseteq
  \mathord{\eventEquivSculpt}$ for any embedding $\embedMorphism:
  Q\hookrightarrow B^k$.
\end{lem}

\begin{proof}%
  The proof is by \textit{reductio ad absurdum}; suppose that there is
  an embedding for which
  $\mathord\sim_n\subseteq \mathord{\eventEquivSculpt}$.  However,
  since according to Proposition~\ref{pr:sts=stpi} $\quotientofwrt{\hintost(Q)}{\eventEquivSculpt}$ produces the
  same labels as $\sculpintost(\sculpture{Q}{\embedMorphism}{\bulk{d}})$, and since
  $\sculpintost$ labels each different cell with a different label
  (because of the injectivity of the embedding), we have a
  contradiction.
\end{proof}

We construct $\sim_{ n+ 1}$ from $\sim_n$ by finding and repairing
\emph{homotopy pairs}, which consist of two paths of the form

\vspace{-2ex}
\begin{equation*}
  \pi \transition{s}e_{1} \transition{t} v_{1}\ \dotsm\
  v_{n-1} \transition{s}e_{n} \transition{t} q
  \qquad \pi \transition{s}e'_{1} \transition{t} v'_{1}\ \dotsm\
  v'_{n-1} \transition{s}e'_{n} \transition{t} q
\ .
\end{equation*}
The shortest homotopy pair is an \emph{interleaving}, a pair of
two transitions.

\begin{lem}
  \label{le:homtp pair}
  If $q\in Q_0$ is such that
  $|\{ \sigma\mid( q, \sigma)\in \rho_n\}|> 1$, then there exists a
  homotopy pair
  with final state $q$.
\end{lem}

\begin{proof}%
  We have
  $|\{ \hintost( \pi)\mid \finishPath{\pi}= q\}|\ge |\{ \sigma\mid( q,
  \sigma)\in \rho_n\}|\ge 2$, hence at least two different rooted paths must
  lead to~$q$. These might share a common prefix $\pi$, which can also be the empty path, \ie~starting at the root.
According to Lemma~\ref{le:hda to sculp zero is enough} we look only at states, and because of Lemma~\ref{lem:CanonicalPath} we can look only at sequential paths, which according to Lemma~\ref{lem:STForSequntialPaths} the corresponding ST-configuration is formed of summing up their events, and since by Lemma~\ref{le:algo no diff numbers} these have the same number of events, the paths have the same length.
\end{proof}

Now if the homotopy pair is an interleaving
$\smash{\pi \transition{s}e^{a} \transition{t} v
  \transition{s}e^{b} \transition{t} q}$,
$\smash{\pi \transition{s}e^{c}}$
$\smash{\transition{t} v' \transition{s}e^{d} \transition{t}
  q}$, %
then we must repair by identifying
$\lambda(e^{a})$ with $\lambda(e^{d})$ and $\lambda(e^{c})$ with $\lambda(e^{b})$.  If it
is not, then there are several choices for identifying events, and
some of them may lead into situations like in Lemma~\ref{le:hda to
  sculpt rho not inj}.  Let $\tau$ be any permutation on
$\{ 1,\dotsc, n\}$ with 
$\tau( 1)\ne 1$ and $\tau( n)\ne n$, %
then we can identify $\lambda(e_i)$ with $\smash{\lambda(e'_{ \tau( i)})}$ 
for all $i= 1,\dotsc, n$. 
The restriction on the permutation is imposed by
the fact that we only identify transitions that can possibly be
concurrent, which is not the case for two transitions starting
from, or ending in, the same cell.

Let $\mathord{\sim_{ n+ 1}}\supsetneq \mathord{\sim_n}$ be the
equivalence relation thus generated which should still be antisymmetric (otherwise choose another permutation).
As this inclusion is proper, it
is clear that the described process either stops with a
Lemma~\ref{le:hda to sculpt rho not inj} situation, which cannot be
resolved without backtracking, or with a relation $\rho_N$ which
satisfies Theorem~\ref{th:sculptalgo}.

\begin{exa}\label{exa:backtrack}
 We give an example to illustrate why backtracking might be necessary
when applying the algorithm.  Figure~\ref{fig:backtrack} is a
variation of the example in Figure~\ref{fig:nosculpt} which, as the
labeling on the top right shows, \emph{can} be sculpted.  However, if
we \emph{start} our procedure by resolving the homotopy pair on the
left in a ``wrong'' way, see the bottom of the figure, then we get
into a contradiction in the top right corner and must backtrack.

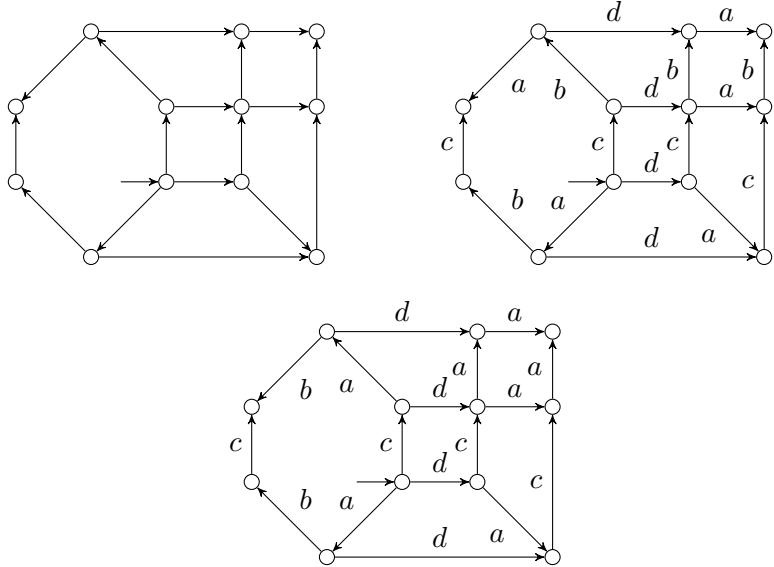
\begin{figure}[t]
  \centering
  \begin{tikzpicture}[>=stealth', x=1cm, y=1cm]
    \node[state] (1) at (1,3) {};
    \node[state] (2) at (3,3) {};
    \node[state] (3) at (4,3) {};
    \node[state] (4) at (0,2) {};
    \node[state] (5) at (2,2) {};
    \node[state] (6) at (3,2) {};
    \node[state] (7) at (4,2) {};
    \node[state] (8) at (0,1) {};
    \node[state, initial] (9) at (2,1) {};
    \node[state] (10) at (3,1) {};
    \node[state] (11) at (1,0) {};
    \node[state] (12) at (4,0) {};
    \path (1) edge (2);
    \path (5) edge (1);
    \path (6) edge (2);
    \path (2) edge (3);
    \path (7) edge (3);
    \path (1) edge (4);
    \path (5) edge (6);
    \path (6) edge (7);
    \path (8) edge (4);
    \path (11) edge (8);
    \path (9) edge (5);
    \path (9) edge (10);
    \path (10) edge (6);
    \path (9) edge (11);
    \path (11) edge (12);
    \path (12) edge (7);
    \path (10) edge (12);
  \end{tikzpicture}
  \hspace{3em}
  \begin{tikzpicture}[>=stealth', x=1cm, y=1cm]
    \node[state] (1) at (1,3) {};
    \node[state] (2) at (3,3) {};
    \node[state] (3) at (4,3) {};
    \node[state] (4) at (0,2) {};
    \node[state] (5) at (2,2) {};
    \node[state] (6) at (3,2) {};
    \node[state] (7) at (4,2) {};
    \node[state] (8) at (0,1) {};
    \node[state, initial] (9) at (2,1) {};
    \node[state] (10) at (3,1) {};
    \node[state] (11) at (1,0) {};
    \node[state] (12) at (4,0) {};
    \path (1) edge node {$d$} (2);
    \path (5) edge node {$b$} (1);
    \path (6) edge node {$b$} (2);
    \path (2) edge node {$a$} (3);
    \path (7) edge node {$b$} (3);
    \path (1) edge node {$a$} (4);
    \path (5) edge node {$d$} (6);
    \path (6) edge node {$a$} (7);
    \path (8) edge node {$c$} (4);
    \path (11) edge node[swap] {$b$} (8);
    \path (9) edge node {$c$} (5);
    \path (9) edge node {$d$} (10);
    \path (10) edge node {$c$} (6);
    \path (9) edge node[swap] {$a$} (11);
    \path (11) edge node {$d$} (12);
    \path (12) edge node {$c$} (7);
    \path (10) edge node[swap] {$a$} (12);
  \end{tikzpicture}

  \bigskip
  \begin{tikzpicture}[>=stealth', x=1cm, y=1cm]
    \node[state] (1) at (1,3) {};
    \node[state] (2) at (3,3) {};
    \node[state] (3) at (4,3) {};
    \node[state] (4) at (0,2) {};
    \node[state] (5) at (2,2) {};
    \node[state] (6) at (3,2) {};
    \node[state] (7) at (4,2) {};
    \node[state] (8) at (0,1) {};
    \node[state, initial] (9) at (2,1) {};
    \node[state] (10) at (3,1) {};
    \node[state] (11) at (1,0) {};
    \node[state] (12) at (4,0) {};
    \path (1) edge node {$d$} (2);
    \path (5) edge node {$a$} (1);
    \path (6) edge node {$a$} (2);
    \path (2) edge node {$a$} (3);
    \path (7) edge node {$a$} (3);
    \path (1) edge node {$b$} (4);
    \path (5) edge node {$d$} (6);
    \path (6) edge node {$a$} (7);
    \path (8) edge node {$c$} (4);
    \path (11) edge node[swap] {$b$} (8);
    \path (9) edge node {$c$} (5);
    \path (9) edge node {$d$} (10);
    \path (10) edge node {$c$} (6);
    \path (9) edge node[swap] {$a$} (11);
    \path (11) edge node {$d$} (12);
    \path (12) edge node {$c$} (7);
    \path (10) edge node[swap] {$a$} (12);
  \end{tikzpicture}
  \caption{%
    \label{fig:backtrack}
    A backtracking example of a sculpture where a homotopy pair is treated.}
\end{figure}

\end{exa}

\begin{rem}\label{rem:algorithm}
In conclusion, our final algorithm has the following steps:
\begin{enumerate}
\item Traverse the HDA using the \hintost, but because of Lemma~\ref{le:hda to sculp zero is enough} we can restrict to only states from $Q_{0}$, and because of Lemmas~\ref{lem:HomotopicPaths} and \ref{lem:HomotopicPathsType} we can look only at sequential paths. This means that applying \hintost\ is like traversing the graph formed of the $Q_{0}\cup Q_{1}$. This forms the $\rho_0$.
\item During the graph traversal, at each state check the Lemma~\ref{le:algo no diff numbers} in constant time.
The algorithm can stop here if the check does not succeed.
\item\label{rem:algorithmMinimalEquiv} Form the minimal equivalence relation on $\universalEvents{Q}$, called $\mathord\sim_1$ which produces the coarser labeling $\rho_1$. This is the same way as the definition of proper event identification starts, i.e., with an equivalence relation on $\universalEvents{Q}$. To build $\mathord\sim_1$ we need to also traverse all the concurrency 2-cells from $Q_{2}$.
\item Check in each state the Lemma~\ref{le:hda to sculpt rho not inj} in constant time.
The algorithm can stop here if the check for $\rho_1$ does not succeed.
\item\label{rem:algorithmHomotopy} Traverse one more time the graph formed of the $Q_{0}\cup Q_{1}$ to find any homotopy pair, as in Lemma~\ref{le:homtp pair}. 
\begin{enumerate}
\item For each homotopy pair add more equivalences to the previous $\rho_n$ resulting from choosing one of the permutations of the transitions of this pair, as explained before.
\item For all states that have their ST-labels changed by the new equated transitions, check again the Lemma~\ref{le:hda to sculpt rho not inj}.
If the check fails, then either backtrack and try another permutation, or the algorithm stops.
\item For each homotopy pair we may need to try out all the possible permutations, which are $(k-1)!$, with $k$ the length of the homotopy pair.
\end{enumerate}

\end{enumerate}
The complexity of this simple algorithm is mostly influenced by the backtracking that needs to be done for each homotopy pair in step (\ref{rem:algorithmHomotopy}).
Therefore, the complexity increases with the number of homotopy pairs that exist in the graph, their respective lengths, and the amount of relabeling which triggers checking of Lemma~\ref{le:hda to sculpt rho not inj}. Note that the less concurrency is in the HDA, the more homotopy pairs might exist, which at the same time reduces the amount of work done in step (\ref{rem:algorithmMinimalEquiv}), since this decreases with the decrease in amount of concurrency. The length of the homotopy pairs contributes the most, since this induces a factorial amount of backtracking. For the minimal homotopy pair of length 2 (the interleaving square) there is only one choice of permutation. Whereas, the worst case is for a 1-dimensional HDA consisting of a single homotopy pair of length $\num{Q_{1}}/2$.
Calculating precisely the complexity is left as future work, the same as finding more efficient algorithms (e.g., how to combine all into a single traversal).
\end{rem}

\section{Euclidean Cubical Complexes are Sculptures}
\label{se:euclid}

This section provides a connection between the combinatorial
intuition of sculptures and the geometric intuition of Euclidean
HDA.
It thus gives a concrete way of identifying precisely the events
that a grid imposes on any of its subsets. This is how several works
on deadlock detection model their studied systems, as ``grids with
holes'', which are geometric sculptures in our terminology.
We give below only strictly necessary definitions, and one is kindly pointed to \cite{Grandis09book, DBLP:books/sp/FajstrupGHMR16}
for background in directed topology.

\paragraph*{Directed topological spaces}

A directed topological space, or \emph{d-space}, is a pair
$( X, \po P X)$ consisting of a topological space $X$ and a set
$\po P X\subseteq X^I$ of \emph{directed paths} in $X$ which contains
all constant paths and is closed under concatenation, monotone
reparametrization, and subpath.

Prominent examples of d-spaces are the directed interval
$\po I=[ 0, 1]$ with the usual ordering and its cousins, the directed
$n$-cubes $\po I^n$ for $n\ge 0$.  Similarly, we have the directed
Euclidean spaces $\po \Real^n$, with the usual ordering, for $n\ge
0$.

Morphisms $f:( X, \po P X)\to( Y, \po P Y)$ of d-spaces are those
continuous functions that are also \emph{directed}, that is, satisfy
$f\circ \gamma\in \po P Y$ for all $\gamma\in \po P X$.  It can be
shown that for an arbitrary d-space $( X, \po P X)$,
$\po P X= X^{ \po I}$.

\paragraph*{Geometric realization}

The \emph{geometric realization} of a precubical set $Q$ is the
d-space
$\georel Q= \bigsqcup_{ n\ge 0} Q_n\times \po I^n / \mathord\sim$,
where the equivalence relation $\sim$ is generated by
\begin{align*}
  ( s_i q,( u_1,\dotsc, u_{ n- 1})) &\sim ( q,( u_1,\dotsc, u_{ i- 1},
  0, u_{ i+ 1},\dotsc, u_{ n- 1}))\,, \\
  ( t_i q,( u_1,\dotsc, u_{ n- 1})) &\sim ( q,( u_1,\dotsc, u_{i- 1},
  1, u_{ i+ 1},\dotsc, u_{ n- 1}))\,.
\end{align*}
(Technically, this requires us to define disjoint unions and quotients
of d-spaces, but there is nothing surprising about these definitions,
see~\cite{Uli05PhD}.)

Geometric realization is naturally extended to \emph{morphisms} of
precubical sets: if $f: Q\to R$ is a precubical morphism, then
$| f|:| Q|\to| R|$ is the directed map given by
$| f|( q,( u_1,\dotsc, u_n))=( f( q),( u_1,\dotsc, u_n))$.  Geometric
realization then becomes a functor from the category of precubical
sets to the category of d-spaces.

\paragraph*{Euclidean Precubical Sets}

Intuitively, a precubical set is Euclidean if its geometric
realization can be embedded into a hypercube lattice in some
$\smash{\po \Real^d}$.  We make this precise below.

\begin{defi}
  \label{de:grid}
  A non-selflinked precubical set $Q$ with $\dim Q= d< \infty$ is a
  \emph{grid} if there exist $M_1,\dotsc, M_d\in \Nat$ and a bijection
  $\Phi:\{ 1,\dotsc, M_1\}\times\dotsm\times\{ 1,\dotsc, M_d\}\to
  Q_d$, such that for all $k\in\{ 1,\dotsc, d\}$ and all
  $( i_1,\dotsc, i_d)\in\{ 1,\dotsc, M_1\}\times\dotsm\times\{
  1,\dotsc, M_{ k- 1}\}\times\{ 1,\dotsc, M_k- 1\}\times\{ 1,\dotsc,
  M_{ k+ 1}\}\times\dotsm\times\{ 1,\dotsc, M_d\}$,
  \begin{equation}
    \label{eq:grid}
    \qquad t_k \Phi( i_1,\dotsc, i_d)= s_k \Phi( i_1,\dotsc, i_{
      k- 1}, i_k+ 1, i_{ k+ 1},\dotsc, i_d)\,,
  \end{equation}
  and there are no other face relations between cubes in $Q$.
\end{defi}

\begin{figure}[tbp]
  \centering
  \begin{tikzpicture}[-, >=stealth']
    \path[->, gray] (-.2,0) edge (5.5,0);
    \path[->, gray] (0,-.2) edge (0,3.5);
    \fill[gray!35] (0,0) -- (4,0) -- (4,2) -- (0,2) -- (0,0);
    \foreach \x in {0,1,2,3,4} {
      \path (\x,0) edge (\x,2);
      \path[gray] (\x,-.1) edge (\x,.1);
      \node[font=\footnotesize] at (\x,-.3) {$\x$};
    }
    \foreach \y in {0,1,2} {
      \path (0,\y) edge (4,\y);
      \path[gray] (-.1,\y) edge (.1,\y);
      \node[font=\footnotesize] at (-.3,\y) {$\y$};
    }
  \end{tikzpicture}
  \caption{A two-dimensional grid.}
  \label{fi:grid}
\end{figure}

Hence a grid is a product\footnote{Technically, a \emph{tensor
    product}, see \cite{Uli05PhD}.} of \emph{long intervals}:
one-dimensional precubical sets with $1$-cells $1,\dotsc, M_j$ which
are connected such that the upper face of $M_i$ is the lower face of
$M_{ i+ 1}$ for all $i= 1,\dotsc, j- 1$.  Figure \ref{fi:grid} shows
an example of a two-dimensional grid with $M_1= 4$ and $M_2= 2$.  The
geometric realization of a grid is a \emph{subdivided cube}: it can be
embedded into $\smash{\po \Real^d}$ as the product of intervals
$[ 0, M_1]\times\dotsm\times[ 0, M_d]$.

\begin{defi}
  A precubical set $Q$ is \emph{Euclidean} if there exists a grid $G$
  and an embedding $Q\hookrightarrow G$.
\end{defi}

Intuitively, these are precisely the ``geometric sculptures'' referred
to in the introduction: subcomplexes of subdivided cubes.  The next
theorem shows that geometric sculptures and combinatorial
sculptures are the same.

\begin{thm}
  \label{th:scu-euc}
  A precubical set can be sculpted iff it is Euclidean.
\end{thm}

\begin{proof}%
  First off, any bulk is a grid, hence any sculpture can be embedded
  into a grid.
  For the reverse direction, it suffices to show that any grid is a
  sculpture.

  Consider a grid of dimension $d$ with $M_{1},\dots,M_{d}$ as the
  number of grid positions in any dimension.  We develop a naming
  scheme using Chu-style labels as in the canonical naming of bulks
  from Section~\ref{sec_sculptures}. We use the following list of
  events:
  $(
  e_{1}^{1},\dots,e_{1}^{M_{1}},e_{2}^{1},\dots,e_{2}^{M_{2}},\dots,$
  $e_{d}^{1},\dots,e_{d}^{M_{d}} )$ and to each event we give values
  from $\{0,\executing ,1\}$. The tuples have dimension
  $m=\sum_{1\leq i\leq d} M_{i}$.  Construct the bulk $\bulk{m}$ of
  dimension $m$ using the canonical naming starting with the m-tuple
  of the events ordered as above containing only $\executing$ values.
  Each $d$-cell of $X_{d}$ is identified by one of the grid cells
  (\ie~the bijection of the grid) as a $d$-tuple of indices
  $(i_{1},\dots,i_{d})$ to which we give an $m$-tuple label
  constructed as follows:
\[
\begin{cases}
    e_{k}^{i} = 1 & \forall i< i_{k} \,, \\
    e_{k}^{i_{k}} = \executing &  \\
    e_{k}^{i} = 0 & \forall i : i_{k} < i \leq M_{k}\,,
\end{cases}
\]
for $1\leq k\leq d$.

We then label all the faces of each $d$-cell with the canonical naming
starting from the above.  This face labeling is consistent with the
face equality restrictions \eqref{eq:grid} of the grid.  Indeed, take
two $d$-cells of the grid that have faces equated, \ie~pick two $d$-tuples
differing in only one index $(\dots,i_{k},\dots)$ and
$(\dots,i_{k}+1,\dots)$ which are named by the $m$-tuples
$(\dots,1,e_{k}^{i_{k}}=\executing,0,\dots)$ respectively
$(\dots,1,1,e_{k}^{i_{k}+1}=\executing,\dots)$ called $q_{d}^{i_{k}}$
respectively $q_{d}^{i_{k}+1}$.  The face maps are named as:
$t_k (q_{d}^{i_{k}})=(\dots,1,e_{k}^{i_{k}}=1,0,\dots)$ and
$s_k (q_{d}^{i_{k}+1})=(\dots,1,1,e_{k}^{i_{k}+1}=0,\dots)$ which are
the same, thus the equality \eqref{eq:grid} is respected.

Using the above naming, it is easy to construct an embedding from the grid $(M_{1},\dots,M_{d})$ into the bulk $\bulk{m}$: it maps each cell of the grid named by some m-tuple into the cell from the bulk that has the same name.
All the cells are uniquely named in the grid, and thus the mapping is correctly defined.
\end{proof}

\section{Conclusion}
\label{se:conc}

Using a precise definition of sculptures as higher-dimensional
automata (HDA), we have shown that sculptures are isomorphic to
regular ST-structures and also to regular Chu spaces.  This nicely
captures Pratt's event-state duality~\cite{Pratt92concur}.  We have
also shown that sculptures are isomorphic to Euclidean cubical
complexes, providing a link between geometric and combinatorial
approaches to concurrency.

We have made several claims in the introduction about HDA that can or
cannot be sculpted.  We sum these up in the next theorem; detailed
proofs are in Appendix~\ref{sec_proofs_conc}.

\begin{thm}
  \label{th:conc}\label{th_non_sculptures}
  \begin{enumerate}
  \item\label{propBrokenBoxGlabbeekExample} There are acyclic HDA
    which cannot be sculpted.
  \item There is an HDA which cannot be sculpted, but whose unfolding
    can be sculpted.
  \item There is an HDA which can be sculpted, but whose unfolding
    cannot be sculpted.
  \item There is an HDA which can be sculpted and whose unfolding can be sculpted.
  \item There is an HDA which cannot be sculpted and whose unfolding
    cannot be sculpted.
  \end{enumerate}
\end{thm}

The HDA from Figs.~\ref{fig:boxes} (right)
and~\ref{fig:speedAngelDemon} are acyclic but cannot be sculpted.  It
is enough to apply the minimal equivalence of the decision algorithm
to obtain two cells with the same ST-label, \cf~Lemma~\ref{le:hda to
  sculpt rho not inj}.  This proves part~(1) of the theorem.

Both these examples are also their own unfoldings, which proves
part~(5).  Part~(2) is proven by the triangle in
Figure~\ref{fig:unfoldings}, which cannot be sculpted due to
Lemma~\ref{le:algo no diff numbers}.  For part~(4) we can use the
triangle's unfolding and the fact that this is its own unfolding.
Part~(3) is proven by Figure~\ref{fig:boxes}.
Finally, also the one-dimensional HDA from Figure~\ref{fig:nosculpt} cannot be sculpted.
There are several interleaving squares (Lemma~\ref{le:homtp pair}), so
the algorithm has to identify all transitions labeled $a$, which leads
to a contradiction {\`a}~la Lemma~\ref{le:hda to sculpt rho not inj}.

\paragraph*{Acknowledgements}

The authors are grateful to Lisbeth Fajstrup, Samuel Mimram and
Emm\-anuel Haucourt for multiple fruitful discussions on the subject
of this paper, and to Martin Steffen and Olaf Owe for help with an
early version of this paper.

\label{thisistheend}

\bibliographystyle{alpha}
\bibliography{mybib}

\newpage
\appendix

\section{Ordered precubical sets}\label{sec_ordered_pcs}

Fix a consistent precubical set $Q$ and an order $\eventless$ on the set $\universalEvents{Q}$ of universal labels of $Q$. For every $n>0$ and every $n$--cell $q\in Q_n$ let $\sigma(q):\{1,\dotsc,n\}\to\{1,\dotsc,n\}$ be the unique permutation such that
	\[
		\lambda_{\sigma(q)(1)}(q)\eventless\lambda_{\sigma(q)(2)}(q)\eventless\dots\eventless\lambda_{\sigma(q)(n)}(q).
	\]
Let $Q'$ be a precubical set that has the same cells as $Q$, \ie\ $Q'_n=Q_n$ for all $n\geq 0$, and face maps given by
	\[
		s'_i(q)=s_{\sigma(q)(i)}(q),\qquad t'_i(q)=t_{\sigma(q)(i)}(q).
	\]
It remains to check that the face maps $s'_i$ and $t'_i$ satisfy the precubical relations.

Define functions $\mathsf{d}_i:\{1,\dotsc,n-1\}\to \{1,\dotsc,n\}$:
\[
	\mathsf{d}_i(k)=\begin{cases}
		k & \text{for $k<i$},\\
		k+1 & \text{for $k\geq i$,}
	\end{cases}
\]
\begin{lem}\label{lem:AltPreCubRel}
	Let $Q_n$, $n\geq 0$ be a family of sets and for every $n$ let  $s_i,t_i:Q_n\to Q_{n-1}$, $i\in\{1,\dotsc,n\}$ be maps. The following conditions are equivalent:
	\begin{itemize} 
		\item
			The maps $s_i, t_i$ satisfy the precubical relations (\ie\ $Q_n$ with the maps $s_i, t_i$ form a precubical set);
		\item
			$\alpha_i\beta_j=\beta_k \alpha_l$ for $\alpha,\beta\in\{s,t\}$ and all integers $i,j,k,l$ such that $\{\mathsf{d}_j(i),j\}=\{\mathsf{d}_l(k),l\}$.
	\end{itemize}
\end{lem}
\begin{proof}
	The latter condition is satisfied only when $(i,j,k,l)=(i,j,i,j)$ (for any $i,j$) or when $(i,j,k,l)=(i,j,j-1,i)\mbox{ or }(j-1,i,i,j)$ (for $i<j$).
\end{proof}

\begin{lem}\label{lem:B}
	For $q\in Q_n$, $\alpha\in\{s,t\}$ and $k\in\{1,\dotsc,n\}$
	\[
		\sigma(q)\circ \mathsf{d}_k=\mathsf{d}_{\sigma(q)(k)}\circ \sigma(\alpha_{\sigma(q)(k)}(q)).
	\]
\end{lem}
\begin{proof}
	Both maps have the same image $\{1,\dotsc,n\}\setminus\{\sigma(q)(k)\}$ so it is enough to show that they both are increasing. The compositions of both sides with $\lambda(q)$, which when seen as a function from $\{1\dots n\}\to \universalEvents{Q}$, it is increasing, are
	\[
		\lambda(q)\circ \sigma(q) \circ \mathsf{d}_k
	\]
	and
	\[
		\lambda(q)\circ\mathsf{d}_{\sigma(q)(k)}\circ \sigma(\alpha_{\sigma(q)(k)}(q))
		=\lambda(\alpha_{\sigma(q)(k)}(q))\circ \sigma(\alpha_{\sigma(q)(k)}(q)).
	\]
        \smallskip

	\noindent They are increasing since $\lambda(x)\circ \sigma(x)$ is increasing for all $x$.
	The equation above follows from Lemma~\ref{le:multilabels}.(1). 
\end{proof}

\begin{lem}
	The maps $s'_i$ and $t'_i$ satisfy the precubical relations.
\end{lem}
\begin{proof}
	We will use the criterion in Lemma \ref{lem:AltPreCubRel}.
	Choose $i,j,k,l$ such that $(\mathsf{d}_j(i),j)=(l,\mathsf{d}_l(k))$. We have
	\[
		\alpha'_i(\beta'_j(q))
		=\alpha'_i(\beta_{\sigma(q)(j)}(q))
		=\alpha_{\sigma(\beta_{\sigma(q)(j)}(q))(i)}(\beta_{\sigma(q)(j)}(q))
	\]
	and
	\[
		\beta'_k(\alpha'_l(q))
		=\beta'_k(\alpha_{\sigma(q)(l)}(q))
		=\beta_{\sigma(\alpha_{\sigma(q)(l)}(q))(k)}(\beta_{\sigma(q)(l)}(q))		
	\]
	Since (Lemma \ref{lem:B})
	\[
		\mathsf{d}_{\sigma(q)(j)}(\sigma(\beta_{\sigma(q)(j)}(q))(i))
		=\sigma(q)(\mathsf{d}_j(i))
		=\sigma(q)(l),
	\]
	and
	\[
		\mathsf{d}_{\sigma(q)(l)}(\sigma(\alpha_{\sigma(q)(l)}(q))(k))
		=\sigma(q)(\mathsf{d}_l(k))
		=\sigma(q)(j),
	\]
	the conclusion follows.
\end{proof}

\section{Proofs for Sec.~\ref{se:conc}}\label{sec_proofs_conc}

\begin{proof}[Proof of Theorem~\ref{th_non_sculptures}]
The two first examples of the theorem are 2-dimensional HDAs which are also their own history unfoldings.

\begin{figure}[tbp]
  \begin{minipage}[b]{.5\linewidth}
  \centering
  \begin{tikzpicture}[>=stealth', x=3cm, y=2cm]
    \begin{scope}[xshift=20em]
      \coordinate (000') at (0,0);
      \coordinate (001') at (1,0);
      \coordinate (010') at (.3,.3);
      \coordinate (100') at (0,1);
      \coordinate (011a') at (1.5,.15);
      \coordinate (011b') at (1.3,.3);
      \coordinate (101') at (1,1);
      \coordinate (110') at (.3,1.3);
      \coordinate (111') at (1.3,1.3);
      \fill[gray!20] (000') -- (001') -- (011a') -- (111') -- (011b') --
      (010') -- (000');
      \fill[gray!35] (000') -- (010') -- (011b') -- (111') -- (101') --
      (100') -- (000');
      \fill[gray!45] (100') -- (101') -- (111') -- (110') -- (100');

      \node[state, initial] (000) at (000') {};
      \foreach \a in {001,010,100,011a,011b,101,110,111}
      \node[state] (\a) at (\a') {};

      \node at (011a.east) {\;\;\;\;$q_0^1$};
      \node at (011b.east) {\;\;\;\;$q_0^2$};

      \path (000) edge node[below] {$q_1^1$} (001);
      \path (000) edge[densely dashed] node[above] {$q_1^2$} (010);
      \path (000) edge (100);
      \path (001) edge node [below] {$q_1^4$} (011a);
      \path (001) edge (101);
      \path (010) edge[densely dashed] (110);
      \path (010) edge[densely dashed] node [above] {$q_1^3$} (011b);
      \path (100) edge (110);
      \path (100) edge (101);
      \path (011a) edge (111);
      \path (011b) edge (111);
      \path (101) edge (111);
      \path (110) edge (111);
    \end{scope}
  \end{tikzpicture}
  \caption{%
    \label{fig:broken2}
    The broken box example of non-sculpture with needed annotations.}
\end{minipage}
\hfill
\begin{minipage}[b]{.4\linewidth}
   \centering
  \begin{tikzpicture}[>=stealth', x=.9cm, y=.9cm]
    \node[state, initial] (9) at (2,1) {};
    \node[state] (5) at (2,2) {};
    \node[state] (10) at (3,1) {};
    \node[state] (6) at (3,2) {};
    \path (9) edge node[left] {$c^1$} (5);
    \path (5) edge node[below] {$b^2$} (6);
    \path (9) edge node[below] {$b^1$} (10);
    \path (10) edge node[left] {$c^2$} (6);

    \node[state] (1) at (1,3) {};
    \node[state] (2) at (3,3) {$q_{0}^{2}$};
    \node[state] (3) at (4,3) {};
    \node[state] (7) at (4,2) {$q_{0}^{1}$};
    \node[state] (11) at (1,0) {};
    \node[state] (12) at (4,0) {};

    \path (11) edge node[left] {$c^4$} (1);
    \path (1) edge node[above] {$b^3$} (2);
    \path (2) edge (3);
    \path (7) edge (3);
    \path (11) edge node[below] {$b^4$} (12);
    \path (12) edge node[right] {$c^3$} (7);

    \path (5) edge node[right] {\;$a^{2}$} (1);
    \path (6) edge node[left] {$a^{1}$} (2);
    \path (6) edge node[below] {$a^{5}$} (7);
    \path (9) edge node[right] {\;$a^{3}$} (11);
    \path (10) edge node[left] {$a^{4}$\;} (12);

    \path (2) edge[-, densely dashed] (7);
  \end{tikzpicture}
  \caption{%
    \label{fig:nosculpt2}
    A one-dimensional acyclic HDA which cannot be sculpted.}
\end{minipage}
\end{figure}

To show that the broken box cannot be sculpted (refer to Figure~\ref{fig:broken2} for annotations) we apply the labeling strategy described in Section~\ref{sec_decision}.
First we apply the unfolding procedure \hintost\ and for the two problematic corner states $q_{0}^{1}$ and $q_{0}^{2}$ we obtain the following ST-configurations $\hintost(\pi_{1})=(\{q_{1}^{1},q_{1}^{4}\},\{q_{1}^{1},q_{1}^{4}\})$ respectively $\hintost(\pi_{2})=(\{q_{1}^{2},q_{1}^{3}\},\{q_{1}^{2},q_{1}^{3}\})$, where $\pi_{1}$ is the lower rooted path ending in $q_{0}^{1}$ and $\pi_{2}$ is the other lower path ending in $q_{0}^{2}$.
    
The second step is to apply the minimal equivalence \eventEquivHDAs, since this is required for any HDA. Applying \eventEquivHDAs\ on our example equates $q_{1}^{1}\eventEquivHDAs q_{1}^{3}$ because of the three squares: front, top, back, which share horizontal faces. (Transitivity of the equivalence was applied.)
The same argument equates $q_{1}^{2}\eventEquivHDAs q_{1}^{4}$, this time going through the squares left-side, top, right-side.

We now see that through $\rho_{\eventEquivHDAs}$ we have labeled both $q_{0}^{1}$ and $q_{0}^{2}$ with the same label $(\{[q_{1}^{2}],[q_{1}^{3}]\},\{[q_{1}^{2}],[q_{1}^{3}]\})$, made of equivalence classes.
However, for a sculpture we cannot have two cells labeled the same.

Showing that the example of Fig.~\ref{fig:speedAngelDemon} is similar and is enough to look at the transitions labeled with $d$. After applying the minimal equivalence the first two lower states where the lower $d$-transitions (call these $q_{1}^{1}$ and $q_{1}^{2}$) end are labeled with $([q_{1}^{1}],[q_{1}^{1}])$ and $([q_{1}^{2}],[q_{1}^{2}])$. But these are equated by the minimal equivalence due to the two squares that share the upper $d$-transition.

The example from Fig.~\ref{fig:nosculpt} is a one-dimensional acyclic HDA that cannot be sculpted (refer to Figure~\ref{fig:nosculpt2} for annotations used in this argument), which also shows that no two-dimensional structure is needed
for things to turn problematic: already in dimension~$1$ there are acyclic HDA which cannot be sculpted.
Our algorithm detects this without using the minimal equivalence \eventEquivHDAs, because this is not applicable for this example.
However, there are several homotopy pairs of length 2, i.e., called interleaving squares.
Each interleaving square forces the equating of their parallel transitions. In this example, the horisontal transitions of the inner interleaving square, as well as the two ones (upper and lower) connected to it equate the four transition cells that we named $b^{1}\sim b^{2}\sim b^{3}\sim b^{4}$. Similarly, we must equate the vertical transitions of the inner interleaving square, and the two (left and right) connected to it, making $c^{1}\sim c^{2}\sim c^{3}\sim c^{4}$.
Now the four outer interleaving squares that we already treated the $b$ and $c$ transitions have in common the parallel transitions labeled by $a^{1}\sim a^{2}\sim a^{3}\sim a^{4}\sim a^{5}$. 
This necessary equivalence 
makes the two states connected with a dashed line to be identified because they now receive the same ST-configuration as label $(\{[a^{1}],[b^{1}],[c^{1}]\},\{[a^{1}],[b^{1}],[c^{1}]\})$, which cannot be. Moreover, there is no backtracking possible because for the interleaving squares there is only one possible way to equate their transitions; unlike for longer homotopy pairs where we can try several possible equating alternatives, as it was the case in Example~\ref{exa:backtrack} with Figure~\ref{fig:backtrack}.
\end{proof}

\end{document}